\documentclass[a4paper,reqno]{amsart} 
\usepackage{amsfonts}
\usepackage{amssymb}
\usepackage{amsthm}
\usepackage{amsmath}
\usepackage{mathrsfs}
\usepackage{dsfont}
\usepackage{stmaryrd} 
\usepackage[english]{babel}
\usepackage[left=3cm,right=3cm,top=3.5cm,bottom=3cm,headsep=0.7cm]{geometry}

\usepackage{pgf,tikz}
\usetikzlibrary{arrows}
\usetikzlibrary{intersections}

\usepackage[scriptsize,bf]{caption}
\usepackage{floatrow}
\usepackage{enumerate}
\usepackage{enumitem}
\usepackage{color}
\usepackage{hyperref}
\usepackage[normalem]{ulem}

\theoremstyle{plain}
\newtheorem{theorem}{Theorem}[section]
\newtheorem{corollary}[theorem]{Corollary}
\newtheorem{lemma}[theorem]{Lemma}
\newtheorem{proposition}[theorem]{Proposition}
\theoremstyle{definition}

\mathchardef\emptyset="001F
\numberwithin{equation}{section}

\newcommand{\N}{\mathbb N}

\newcommand{\e}{\varepsilon}

\newcommand{\R}{\mathbb R}
\newcommand{\Z}{\mathbb Z}

\newcommand{\C}{\mathbb C}

\renewcommand{\S}{\mathcal S}

\newcommand{\step}[1]{\noindent \textit{Step} #1.}

\newcommand{\mres}{\mathbin{\vrule height 1.6ex depth 0pt width 0.13ex\vrule height 0.13ex depth 0pt width 1.3ex}}
\newcommand{\wstar}{\stackrel{*}\rightharpoonup}

\newcommand{\geo}{\mathrm{d}_{\mathbb{S}^1}}

\newcommand{\D}{\mathrm{D}}
\renewcommand{\d}{\, \mathrm{d}}
\newcommand{\RR}{\mathbb{R}}
\renewcommand{\SS}{\mathbb{S}}
\newcommand{\sm}{\setminus}

\renewcommand{\H}{\mathcal{H}}
\newcommand{\x}{{\times}}
\newcommand{\ol}{\overline}
\newcommand{\ZZ}{\mathbb{Z}}
\newcommand{\de}{\partial}
\newcommand{\NN}{\mathbb{N}}

\title[Coarse graining and large-$N$ behavior of the $d$-dimensional $N$-clock model]{Coarse graining and large-$N$ behavior \\
of the $d$-dimensional $N$-clock model}
\author[M. Cicalese]{Marco Cicalese}
\author[G. Orlando]{Gianluca Orlando}
\author[M. Ruf]{Matthias Ruf}
\address[Marco Cicalese]{Technische Unversit\"at M\"unchen, Zentrum Mathematik, Munich, Germany}
\email{cicalese@ma.tum.de}
\address[Gianluca Orlando]{Technische Unversit\"at M\"unchen, Zentrum Mathematik, Munich, Germany}
\email{orlando@ma.tum.de}
\address[Matthias Ruf]{Ecole polytechnique f\'ed\'erale de Lausanne, SB MATH MATH-GE, Lausanne, Switzerland}
\email{matthias.ruf@epfl.ch}

\begin{document}
\selectlanguage{english}

\begin{abstract}
	We study the asymptotic behavior of the $N$-clock model, a nearest neighbors ferromagnetic spin model on the $d$-dimensional cubic $\varepsilon$-lattice in which the spin field is constrained to take values in a discretization $\mathcal{S}_N$ of the unit circle~$\mathbb{S}^{1}$ consisting of $N$ equispaced points. Our $\Gamma$-convergence analysis consists of two steps: we first fix $N$ and let the lattice spacing $\varepsilon \to 0$, obtaining an interface energy in the continuum defined on piecewise constant spin fields with values in $\mathcal{S}_N$; at a second stage, we let $N \to +\infty$. The final result of this two-step limit process is an anisotropic total variation of $\mathbb{S}^1$-valued vector fields of bounded variation.
\end{abstract}

\maketitle
 

\vspace{1em}


\setcounter{tocdepth}{1}
\tableofcontents

\section{Introduction}

In this paper we are interested in the variational analysis of the $N$-clock model (also known as planar Potts model or $\ZZ_N$-model) in the $d$-dimensional setting. The $N$-clock model is a nearest neighbors ferromagnetic spin model on the cubic lattice in which the spin field is constrained to take values in a set of $N$ equispaced points of the unit circle~$\SS^{1}$. It plays a fundamental role in understanding phase transition phenomena in the theory of classical ferromagnetic spin fields, as it is closely related to the $XY$ (planar rotator) model, for which the spin field is allowed to attain all the values of $\SS^{1}$. In fact, the $N$-clock model is considered as an approximation of the $XY$ model, as for $N$ large enough it predicts Berezinskii-Kosterlitz-Thouless transitions~\cite{FroSpe}, i.e., phase transitions mediated by the formation and interaction of topological singularities, the so-called vortices~\cite{Ber, Kos, KosTho}.

With the aim of describing the relation between the $N$-clock model and the~$XY$ model, probabilistic methods have been used in \cite{vEKO, KO}, while a variational analysis at zero temperature has been only very recently carried out in~\cite{COR1, COR2}. There the authors study the effective behavior of (suitably rescaled versions of) the energy of the~\mbox{$N$-clock} model on the  2-dimensional square lattice $\e \ZZ^2$, examining the case when the number $N = N_\e$ of equi-spaced points on $\SS^{1}$ depends on $\e$ and diverges as $\e \to 0$. The coarse grained model, which describes the microscopic/mesoscopic geometry of the spin field, is strongly affected by the rate of divergence of $N_\e \to +\infty$ as $\e \to 0$. 

In this paper we advance the variational analysis of the $N$-clock model by considering the model on a $d$-dimensional cubic lattice $\e \ZZ^d$, with $d \geq 2$, in the case where the number~$N$ is fixed and independent of $\e$. We shall first identify the limit of the $N$-clock model as $\e \to 0$ keeping $N$ fixed and, at a second stage, we will let $N \to +\infty$. In contrast to the energy of the $XY$ model, the energy resulting from this two-step limit process is by nature unfit to describe the concentration of energy around vortex-like singularities, indicating that the dependence of $N$ on $\e$ seems inevitable with the intent to approximate the $XY$ model at zero temperature. To the best of our knowledge, the explicit identification of the limit energies in the $\e \to 0$ and $N\to\infty$ regimes and in any dimension makes the result contained in this paper the first quantitative answer to the question whether the $N$-clock model approximates the $XY$ model at zero temperature. We shall see that the result is rather analogous to the limiting energy of the $N_\e$-clock model in a specific rate of divergence $N_\e \to +\infty$, chosen among those examined in the two-dimensional setting in~\cite{COR2}. To present in detail the results in this paper, we first summarize the analysis of~\cite{COR2}, starting with some notation. 

Given $N \in \NN$, we consider the set of $N$ equispaced points on the unit circle
\begin{equation*}
	\S_N := \{ \exp\big(\iota \tfrac{2\pi}{N} k \big) \ :  \ k = 0, \dots, N-1\} \, ,
\end{equation*}
where $\iota$ is the imaginary unit. Given an open set $\Omega \subset \RR^2$, the energy associated to an admissible spin field $u \colon \e \ZZ^2 \to \S_{N_\e}$ is given by
\begin{equation*}
	E_\e^{N_\e}(u) := \frac{1}{2} \sum_{\langle i,j \rangle \, \text{in} \, \Omega} \e^2 |u(\e i) - u(\e j)|^2,
\end{equation*}
where the sum is taken over ordered pairs of nearest neighbors $\langle i, j \rangle$, i.e., $(i,j) \in \ZZ^2\x\ZZ^2$ such that $|i-j|=1$ and $\e i, \e j \in \Omega$. We recall that a wide range of phenomena has been observed in~\cite{COR1,COR2} when exploring the possible regimes of $N_\e$. Here we outline the one pertaining to the discussion in the present paper, namely $N_\e \ll \frac{1}{\e |\log \e|}$. The relevant scaling of the energy in this regime is $\frac{N_\e}{2 \pi \e} E_\e^{N_\e}$, sequences of spin fields $u_\e$ with equibounded energy accumulate to vector fields in $BV(\Omega;\SS^1)$, and the scaled energy $\frac{N_\e}{2 \pi \e} E_\e^{N_\e}$ approximates an anisotropic total variation for maps in $BV(\Omega;\SS^1)$. 

In the next theorem we state the result in the regime $N_\e \ll \frac{1}{\e |\log \e|}$ rigorously. We denote by $|\, \cdot \,|_1$ the 1-norm on vectors, by $| \, \cdot \, |_{2,1}$ the anisotropic norm on matrices given by the sum of the Euclidean norms of the columns, and by $\geo$ the geodesic distance on $\SS^1$. For the notation concerning functions of bounded variation we refer to Subsection~\ref{sec:constrained problems}.


\begin{theorem}\cite{COR2} \label{thm:old paper}
	Let $\Omega \subset \RR^2$ be a bounded, open set with Lipschitz boundary. Assume that $N_\e \ll \frac{1}{\e |\log \e|}$. Then the following results hold true: 
	\begin{itemize}[leftmargin=*]
		\item[i)] (Compactness) Let $u_\e \colon \Omega \cap \e \ZZ^2 \to \S_{N_\e}$ be such that $\frac{N_\e}{2\pi \e} E^{N_\e}_\e(u_\e) \leq C$. Then there exists a subsequence (not relabeled) and a function $u \in BV(\Omega;\SS^1)$ such that $u_\e \to u$ in $L^1(\Omega;\RR^2)$. 
		\item[ii)] ($\, \Gamma$-liminf inequality) Assume that $u_\e \colon \Omega \cap \e \ZZ^2 \to \S_{N_\e}$ and $u \in BV(\Omega;\SS^1)$ satisfy $u_\e \to u$ in $L^1(\Omega;\RR^2)$. Then 
			\begin{equation*}
				\liminf_{\e \to 0} \frac{N_\e}{2\pi \e } E^{N_\e}_\e(u_\e)  \geq  \int_{\Omega}{|\nabla u|_{2,1} }{\d x} + |\D^{(c)} u|_{2,1}(\Omega) + \int_{\Omega\cap J_u}{\geo(u^-,u^+)|\nu_{u}|_1}{\d \H^1}  \, .
			\end{equation*} 
		\item[iii)] ($\, \Gamma$-limsup inequality) Let $u \in BV(\Omega;\SS^1)$. Then there exists a sequence $u_\e \colon \Omega \cap \e \ZZ^2 \to \S_{N_\e}$ such that $u_\e \to u$ in $L^1(\Omega;\RR^2)$ and
		\begin{equation*}
				\limsup_{\e \to 0} \frac{N_\e}{2\pi \e} E^{N_\e}_\e(u_\e) \leq \int_{\Omega}{|\nabla u|_{2,1} }{\d x} + |\D^{(c)} u|_{2,1}(\Omega) + \int_{\Omega\cap J_u}{\geo(u^-,u^+)|\nu_{u}|_1}{\d \H^1}.
		\end{equation*}  
	\end{itemize}
\end{theorem}

We are now in a position to present the two main results in this paper. We shall consider $\Omega \subset \RR^d$ a bounded, open set with Lipschitz boundary and the energy defined for admissible spin fields on the $d$-dimensional cubic lattice $u \colon \Omega \cap \e \ZZ^d \to \S_N$ by 
\begin{equation*}
	E_\e^{N}(u) := \frac{1}{2} \sum_{\langle i,j \rangle \, \text{in}\, \Omega} \e^d |u(\e i) - u(\e j)|^2,
\end{equation*}
where the sum is taken over ordered pairs of nearest neighbors $\langle i, j \rangle$, i.e., $(i,j) \in \ZZ^d\x\ZZ^d$ such that $|i-j|=1$ and $\e i, \e j \in \Omega$ (the factor $\frac{1}{2}$ accounts for the fact that each pair is counted twice).  We state the first result concerning the limit of $E_\e^N$ as $\e \to 0$. For $N$ fixed, the physical system is expected to behave like a classical Ising-type system with $N$ phases. (See also~\cite{CafDLL05, AliBraCic06, AliCicSig12, BraPia12, AliCicRuf15, CicSol15, BraCic17, BraKre18, CicForOrl19, BraPia20} for the analysis of spin systems in the surface scaling.) According to the results proven for the Ising system, we expect the limit energy to be finite on functions of bounded variation with values in the finite set $\S_N$. In the next theorem we identify precisely the surface energy concentrated on the interfaces between the phases of the spin field. We denote by $\theta_N := \frac{2\pi}{N}$ the smallest angle between two different vectors in $\S_N$. 

\begin{theorem}[Limit as $\e \to 0$]\label{thm:eps to 0}
	Let $\Omega \subset \RR^d$ be a bounded, open set with Lipschitz boundary. Let $N \geq 2$ and $\theta_N := 2\pi/N$. Then the following results hold true: 
	\begin{itemize}[leftmargin=*]
		\item[i)] (Compactness) Let $u_\e \colon \Omega \cap \e \ZZ^d \to \S_{N}$ be such that $\frac{N}{2\pi \e} E^{N}_\e(u_\e) \leq C$. Then there exists a subsequence (not relabeled) and a function $u \in BV(\Omega;\S_N)$ such that $u_\e \to u$ in $L^1(\Omega;\RR^2)$ as $\e \to 0$. 
		\item[ii)] ($\, \Gamma$-liminf inequality) Assume that $u_\e \colon \Omega \cap \e \ZZ^d \to \S_{N}$ and $u \in BV(\Omega;\S_N)$ satisfy $u_\e \to u$ in $L^1(\Omega;\RR^2)$  as $\e \to 0$. Then 
			\begin{equation*}
				\liminf_{\e \to 0} \frac{N}{2\pi \e } E^{N}_\e(u_\e) \geq \frac{4\sin^2\left(\tfrac{\theta_N}{2}\right)}{\theta_N^2}\int_{\Omega\cap J_u}\geo(u^-,u^+)|\nu_u|_1\,\mathrm{d}\mathcal{H}^{d-1}\, .
			\end{equation*} 
		\item[iii)] ($\, \Gamma$-limsup inequality) Let $u \in BV(\Omega;\S_N)$. Then there exists a sequence $u_\e \colon \Omega \cap \e \ZZ^d \to \S_N$ such that $u_\e \to u$ in $L^1(\Omega;\RR^2)$  as $\e \to 0$ and
		\begin{equation*}
				\limsup_{\e \to 0} \frac{N}{2\pi \e} E^{N}_\e(u_\e) \leq \frac{4 \sin^2\left(\tfrac{\theta_N}{2}\right)}{\theta_N^2}\int_{\Omega\cap J_u}\geo(u^-,u^+)|\nu_u|_1\,\mathrm{d}\mathcal{H}^{d-1}.
		\end{equation*}  
	\end{itemize}
\end{theorem}

To clarify the expression of the limit functional in Theorem~\ref{thm:eps to 0}, we sketch here the proof of the $\Gamma$-limsup inequality in a very simple setting. Assume that $\Omega$ is the unit cube $Q = (-1/2,1/2)^d$ and~$u$ is the pure-jump function with constant value $u^- = (1,0)$ in $Q^- = (-1/2,1/2)^{d-1} \x (-1/2,0)$ and constant value $u^+ = \exp(\iota k^+ \theta_N)$ in $Q^+ = (-1/2,1/2)^{d-1} \x (0,1/2)$, where $k^+ \in \NN$ is such that $0 \leq k^+ \theta_N \leq \pi$. In this case, the jump set is given by $J_u = (-1/2,1/2)^{d-1} \x \{0\}$. Then~$u_\e$ is constructed by rotating $k^+$ times of an angle $\theta_N$ starting from $u^-$ up to $u^+$ on hyperplanes parallel to the jump set, cf.~Figure~\ref{fig:transition easy case}. More precisely, for $0 \leq k \leq k^+$  we define 
\begin{equation*}
	u_\e(\e i) := \exp(\iota k \theta_N) \quad \text{if} \quad \e i \cdot e_d = k \e 
\end{equation*}
and we put $u_\e(\e i) = (1,0)$ if $\e i \cdot e_d < 0$ and $u_\e(\e i) = \exp(\iota k^+ \theta_N)$ if $\e i \cdot e_d > k^+ \e$, instead. Between two hyperplanes there are $\frac{1}{\e^{d-1}}$ interacting pairs of nearest neighbors. For two such points $\e i, \e j$ we have by a simple geometric argument	$|u_\e(\e i) - u_\e(\e j)| = 2 \sin  (\frac{\theta_N}{2} )$. Summing over all interactions we conclude that 
\begin{equation*}
	\frac{N}{2\pi \e } E_\e^{N}(u_\e) = \frac{1}{2 \theta_N} \sum_{\langle i,j \rangle \, \text{in}\, Q} \e^{d-1} |u_\e(\e i) - u_\e(\e j)|^2 = \frac{1}{\theta_N} \sum_{k=0}^{k^+} 4 \sin^2 \Big(\frac{\theta_N}{2}\Big) = \frac{4\sin^2\big(\tfrac{\theta_N}{2}\big)}{\theta_N^2} k^+ \theta_N  \, .
\end{equation*}
Since $k^+ \theta_N = \geo(u^-,u^+)$, the previous expression reduces to the one in Theorem~\ref{thm:eps to 0} and makes clear the role of $4\sin^2 (\tfrac{\theta_N}{2} )/\theta_N^2$: it is the correcting factor which allows us to pass from the Euclidean distance between vectors to their geodesic distance. The proof of the upper bound is based on the construction in a more general setting  of a recovery sequence which mimics the one presented here in the introduction, cf.~Proposition~\ref{p.Gamma-limsup}. The proof of the lower bound is based on Lemma~\ref{l.sinuslemma}, which shows that the behavior described above is always the most convenient from an energetical point of view.

\begin{figure}[H]
	\includegraphics{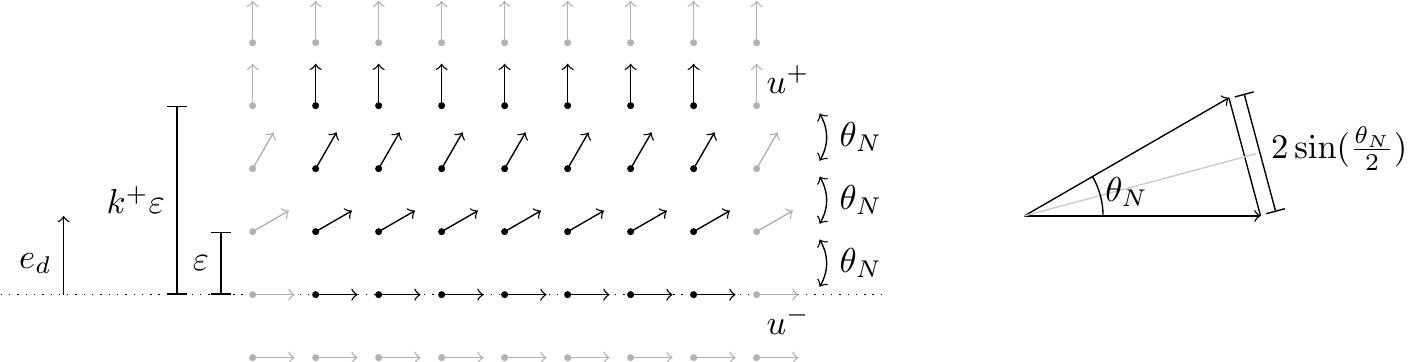}
	\caption{On the left: a recovery sequence in the case of a jump set aligned with the lattice. The spin makes a transition from $u^-$ to $u^+$ jumping with the smallest possible non-zero angle $\theta_N$. On the right: Euclidean distance between two vectors of length 1 with angle $\theta_N$ between them.}
	\label{fig:transition easy case}
\end{figure}

In Section~\ref{sec:constrained problems} we also study the $\Gamma$-convergence of the functionals $E^N_\e$ as $\e \to 0$ under volume constraints on the phases of the spin fields or under Dirichlet boundary conditions.

We are now interested in the limit as $N \to +\infty$ of the energy defined by 
\begin{equation*}
	E_N(u) := \frac{4 \sin^2\left(\tfrac{\theta_N}{2}\right)}{\theta_N^2}\int_{\Omega\cap J_u}\geo(u^-,u^+)|\nu_u|_1\,\mathrm{d}\mathcal{H}^{d-1} , \quad \text{for } u \in BV(\Omega;\S_N) \, ,
\end{equation*}
where $\theta_N := 2\pi/N$, i.e., the energy resulting from the limit process $\e \to 0$ in Theorem~\ref{thm:eps to 0}. Up to the factor $4\frac{\sin^2 \big(\tfrac{\theta_N}{2} \big)}{\theta_N^2}$, which is close to 1 for $N$ large, the energy $E_N$ coincides (for $d=2$) with the limiting energy of Theorem~\ref{thm:old paper} restricted to Caccioppoli partitions taking values in $\S_N$. In the second result of this paper we show that the $\Gamma$-limit of $E_N$ as $N \to +\infty$ agrees with the limiting energy of Theorem~\ref{thm:old paper}. This is rigorously proved in the next theorem, which holds for any dimension $d$.

\begin{theorem}[Limit as $N \to +\infty$]\label{thm:N to infty}
	Let $\Omega \subset \RR^d$ be a bounded, open set with Lipschitz boundary. Then the following results hold: 
	\begin{itemize}[leftmargin=*]
		\item[i)] (Compactness) Let $u_N \colon \Omega \to \S_{N}$ be such that $E_N(u_N) \leq C$. Then there exists a subsequence (not relabeled) and a function $u \in BV(\Omega;\SS^1)$ such that $u_N \to u$ in $L^1(\Omega;\RR^2)$ as $N \to +\infty$. 
		\item[ii)] ($\, \Gamma$-liminf inequality) Assume that $u_N \colon \Omega \to \S_{N}$ and $u \in BV(\Omega;\SS^1)$ satisfy $u_N \to u$ in $L^1(\Omega;\RR^2)$ as $N \to +\infty$. Then 
			\begin{equation*}
				\liminf_{N \to +\infty}  E_N(u_N) \geq \int_{\Omega}{|\nabla u|_{2,1} }{\d x} + |\D^{(c)} u|_{2,1}(\Omega) + \int_{\Omega\cap J_u}{\geo(u^-,u^+)|\nu_{u}|_1}{\d \H^{d-1}}.
			\end{equation*} 
		\item[iii)] ($\, \Gamma$-limsup inequality) Let $u \in BV(\Omega;\SS^1)$. Then there exists a sequence $u_N \colon \Omega \to \S_N$ such that $u_N \to u$ in $L^1(\Omega;\RR^2)$ as $N \to +\infty$ and
		\begin{equation*}
				\limsup_{N \to +\infty} E_N(u_N) \leq \int_{\Omega}{|\nabla u|_{2,1} }{\d x} + |\D^{(c)} u|_{2,1}(\Omega) + \int_{\Omega\cap J_u}{\geo(u^-,u^+)|\nu_{u}|_1}{\d \H^{d-1}}.
		\end{equation*}  
	\end{itemize}
\end{theorem}

The proof of the upper bound in Theorem~\ref{thm:N to infty} is based on the following remark: a map $u \in BV(\Omega;\SS^1)$ can be approximated in energy by maps $W^{1,1}(\Omega;\SS^1)$ which are smooth outside manifolds of codimension $2$; such maps can be suitably sampled far from the singularities to define a $u_N \in BV(\Omega;\S_N)$; a crucial observation is that the precise definition of $u_N$ close to the singularities is not important, as the energy $E_N(u_N)$ does not concentrate close to manifolds of codimension 2. It is worth noticing that the latter feature is peculiar of this regime: in the other regimes studied in~\cite{COR1} where $N = N_\e$ depends on $\e$ and $N_\e \gg \frac{1}{\e |\log \e|}$ the behavior of the recovery sequence around the singularities becomes relevant and makes the generalization to the $d$-dimensional setting of the results in~\cite{COR1} more delicate and out of the scope of the present paper.

\section{Notation and preliminary results}
Let $\SS^{d-1}=\{x\in\R^d:\,|x|=1\}$ be the unit sphere. If $u, v \in \mathbb{S}^1$, their geodesic distance on~$\mathbb{S}^1$ is denoted by $\geo(u,v)$. It is given by the angle in $[0,\pi]$ between the vectors $u$ and $v$, i.e., $\geo(u,v) = \arccos(u\cdot v)$. Observe that
\begin{equation} \label{eq:geo and eucl}
\tfrac{1}{2}|u - v| = \sin\big( \tfrac{1}{2}\geo(u,v) \big) \, .
\end{equation}
We denote the imaginary unit by $\iota$. When it is convenient we will tacitly identify $\RR^2$ with the complex plane~$\C$. Given a vector $a = (a_i)_{i=1}^d \in \RR^d$, its $1$-norm is $|a|_1 = \sum_{i=1}^d |a_i|$. We define the~$(2,1)$-norm of a matrix $A = (a_{ij})_{i,j=1}^d \in \RR^{d \times d}$ as the sum of the Euclidean norms of its columns,~i.e., 
\begin{equation*}
|A|_{2,1} := \sum_{j=1}^d \Big( \sum_{i=1}^d |a_{ij}|^2 \Big)^{\! 1/2} .
\end{equation*}
Given a unit vector $\nu\in \SS^{d-1}$, we denote by $Q_\nu$ a cube with two sides orthogonal to $\nu$, namely, we consider an orthonormal basis $(\nu,\nu_2,\dots,\nu_d)$ of $\R^d$ and we define
\begin{equation}\label{eq:defQ_nu}
Q_{\nu}=\left\{x\in\R^d:\, | x \cdot \nu |<\tfrac{1}{2} \, , \ |  x \cdot \nu_i |<\tfrac{1}{2}\right\} \, .
\end{equation}
For two sequences $\alpha_{\e}$ and $\beta_{\e}$ of positive numbers, we write $\alpha_{\e}\ll \beta_{\e}$ if $\lim_{\e \to 0}\tfrac{\alpha_{\e}}{\beta_{\e}}=0$.

\subsection{BV-functions}\label{subsec:BV}
In this section we recall basic facts about functions of bounded variation. For more details we refer to the monograph \cite{AmbFusPal}.

Let $O\subset\RR^d$ be an open set. A function $u\in L^1(O;\RR^n)$ is a function of bounded variation if its distributional derivative $\D u$ is given by a finite matrix-valued Radon measure on $O$. In that case, we write $u\in BV(O;\RR^n)$.	
\\
The space $BV_{{\rm loc}}(O;\mathbb{R}^n)$ is defined as usual. The space $BV(O;\RR^n)$ becomes a Banach space when endowed with the norm $\|u\|_{BV(O)}=\|u\|_{L^1(O)}+|\D u|(O)$, where $|\D u|$ denotes the total variation measure of $\D u$. The total variation with respect to the anisotropic norm $| \cdot |_{2,1}$ is denoted by $|\D u|_{2,1}$. When $O$ is a bounded Lipschitz domain, then $BV(O;\RR^n)$ is compactly embedded in $L^1(O;\RR^n)$. We say that a sequence $u_n$ converges weakly$^*$ in $BV(O;\RR^n)$ to $u$ if $u_n\to u$ in $L^1(O;\RR^n)$ and $\D u_n\overset{*}{\rightharpoonup}\D u$ in the sense of measures.

We state some fine properties of $BV$-functions. To this end, we need some definitions. A function $u\in L^1(O;\mathbb{R}^n)$ is said to have an approximate limit at $x\in O$ whenever there exists $z\in\mathbb{R}^n$ such that
\begin{equation*}
\lim_{\rho\to 0}\frac{1}{\rho^d}\int_{B_{\rho}(x)}|u(y)-z|\,\mathrm{d}y=0\,.
\end{equation*}
Next we introduce so-called approximate jump points. Given $x\in O$ and $\nu\in \SS^{d-1}$ we set
\begin{equation*}
B^{\pm}_{\rho}(x,\nu)=\{y\in B_{\rho}(x):\;\pm (y-x) \cdot \nu >0\}\,.
\end{equation*}
We say that $x\in O$ is an approximate jump point of $u$ if there exist $a\neq b\in\mathbb{R}^n$ and $\nu\in \SS^{d-1}$ such that
\begin{equation*}
\lim_{\rho\to 0}\frac{1}{\rho^d}\int_{B_{\rho}^+(x,\nu)}|u(y)-a|\,\mathrm{d}y=\lim_{\rho\to 0}\frac{1}{\rho^d}\int_{B^-_{\rho}(x,\nu)}|u(y)-b|\,\mathrm{d}y=0 \, .
\end{equation*}
The triplet $(a,b,\nu)$ is determined uniquely up to the change to $(b,a,-\nu)$. We denote it by $(u^+(x),u^-(x),\nu_u(x))$ and we let $J_u$ be the set of approximate jump points of $u$. The triplet $(u^+,u^-,\nu_u)$ can be chosen as a Borel function on the Borel set $J_u$. Denoting by $\nabla u$ the approximate gradient of $u$, we can decompose the measure $\D u$ as the sum
\begin{equation*}
\D u(B)=\int_B\nabla u\,\mathrm{d}x+\int_{J_u\cap B}(u^+ -u^-)\otimes\nu_u \,\mathrm{d}\mathcal{H}^{d-1}+\D^{(c)}u(B) \, ,
\end{equation*}
where $\D^{(c)}u$ is the so-called Cantor part and $\D^{(j)}u = (u^+ - u^- )\otimes\nu_u \mathcal{H}^{d-1} \mres J_u$ is the so-called jump part.
If $\S \subset \R^n$, we define the space $BV(O;\S)$ as the space of those functions $u\in BV(O;\mathbb{R}^n)$ such that $u(x)\in \S$ for $\mathcal{L}^d$-a.e.\ $x \in O$.

We will need the slicing properties of $BV$-functions. Given a unit vector $\xi \in \SS^{d-1}$, we denote by $\Pi^\xi$ the hyperplane orthogonal to $\xi$. For every set $E \subset \RR^d$ and $z \in \Pi^\xi$, the section of $E$ corresponding to $z$ is the set $E^\xi_z := \{ t \in \RR \ : \ z + t \xi \in E \}$. Accordingly, for any function $u \colon E \to \RR^n$, the function $u^\xi_z \colon E^\xi_z \to \RR^n$ is defined by $u^\xi_z(t) := u(z + t \xi)$. 

We recall a characterization of $BV$ functions by slicing \cite[Remark~3.104]{AmbFusPal}. Let us fix an open set $O\subset \RR^d$ and $u \in L^1(O; \RR^n)$. Then $u \in BV(O; \RR^n)$ if and only if for every $\xi \in \SS^{d-1}$ we have $u^\xi_z \in BV(O^\xi_z;\RR^n)$ for $\H^{d-1}$-a.e.\ $z \in \Pi^\xi$ and 
\[
\int_{\Pi^\xi}{|\D u^\xi_z|(O^\xi_z)}{ \d \H^{d-1}(z)} < \infty \, .
\]
Moreover it is possible to reconstruct the distributional gradient $\D u$ from the gradients of the slices $\D u^\xi_z$ through the formula $\D u \, \xi = \H^{d-1}\mres\Pi^\xi \otimes \D u^\xi_z$, i.e., 
\[
\D u \, \xi(B) = \int_{\Pi^\xi}{\D u^\xi_z(B^\xi_z)}{\d \H^{d-1}(z)} \, ,
\]
for every Borel set $B \subset \RR^d$. More precisely, the same decomposition holds true for each part of the decomposition of~$\D u$, namely
\begin{align*}
\int_{B}{\nabla u \, \xi}{\d x} & = \int_{\Pi^\xi}{\nabla u^\xi_z(B^\xi_z)}{\d \H^{d-1}(z)} \, ,\\
\D^{(c)} u \, \xi(B) & = \int_{\Pi^\xi}{\D^{(c)} u^\xi_z(B^\xi_z)}{\d \H^{d-1}(z)} \, , \\
\D^{(j)} u \, \xi(B) & = \int_{\Pi^\xi}{\D^{(j)} u^\xi_z(B^\xi_z)}{\d \H^{d-1}(z)} \, ,
\end{align*}
for every Borel set $B \subset \RR^d$. Moreover, $J_{u^\xi_z} = (J_u)^\xi_z$ for $\H^{d-1}$-a.e.\ $z \in \Pi^\xi$ and $(u^\xi_z)^{\pm}(t) = (u^\pm)^\xi_z(t)$  ($ = (u^\mp)^\xi_z(t)$, respectively) for every $t \in (J_u)^\xi_z$ if $\xi \cdot \nu_u(z+t\xi) > 0$ (if $\xi \cdot \nu_u(z+t\xi) < 0$, respectively).

\subsection{Known results for general models with finite phases}\label{sec:known}
We recall here some results that were proved for more general energies defined for functions taking values in a given finite set. In~\cite{Bra-Cic-Ruf}, Braides together with the first and third author consider energies $\mathcal{E}_{\e}$ defined for spin variables $u\colon \e\mathcal{L}\to\mathcal{S}$, where $\mathcal{S}$ is a finite set and $\mathcal{L}$ is a so-called thin stochastic lattice. In general, these points sets are located in a fixed neighborhood of a lower-dimensional subspace such that there is a minimal distance between points and there are no arbitrarily large holes in the neighborhood of the subspace. The energies in \cite{Bra-Cic-Ruf} can be of the form
\begin{equation*}
\mathcal{E}_{\e}(u)=\sum_{(\e x,\e y)\in(\e\mathcal{L}\cap \Omega)^2}\e^{d-1}f(x-y,u(\e x),u(\e y)),
\end{equation*}
where the energy density $f\colon \R^d\times\mathcal{S}^2\to [0,+\infty)$ has to satisfy certain growth and decay conditions. We do not state them explicitly here, but we mention that they cover in particular the case when $\mathcal{L}=\Z^d$ is a periodic lattice that is completely contained in the subspace $\R^d$ and 
\begin{equation*}
f(x,m_1,m_2)=\begin{cases}
c\,|m_1-m_2|^2 &\text{if } |x|=1 \, ,
\\
0 &\text{otherwise.}
\end{cases}
\end{equation*}
With $c=\tfrac{N}{4\pi}$ and $\mathcal{S}=\S_N$ we recover the energy $\tfrac{N}{2\pi\e}E_{\e}^N$, so that all results of \cite{Bra-Cic-Ruf} can be applied. In particular, we can use an integral representation result and the characterization of the corresponding integrand through an asymptotic cell formula. Indeed, by \cite[Theorem 5.8]{Bra-Cic-Ruf} we know that in the case of spatially homogeneous interactions the $\Gamma$-limit as $\e\to 0$ of $\tfrac{N}{2\pi\e}E_{\e}^N$ exists, is finite only on $BV(\Omega;\S_N)$, and for $u \in BV(\Omega;\S_N)$ it is of the form
\begin{equation} \label{eq:integral representation}
\int_{\Omega\cap J_u}\varphi(u^-,u^+,\nu_u)\,\mathrm{d}\mathcal{H}^{d-1},
\end{equation}
where the integrand is given by an asymptotic minimization problem with a suitable boundary conditions. More precisely, denoting by $u^{s,r}_{\nu}:\R^d\to\R$ ($\nu\in \mathbb{S}^{d-1}$ and $s,r\in\S_N$) the function
\begin{equation*}
u^{s,r}_{\nu}(x)=\begin{cases}
s &\mbox{if }   x \cdot \nu  > 0 \, ,
\\
r &\mbox{if }  x \cdot \nu  \leq 0 \, ,
\end{cases}
\end{equation*}
then in the case of just nearest neighbor interactions the function $\varphi(s,r,\nu)$ is given by
\begin{equation}\label{eq:cellformula}
\varphi(s,r,\nu)=\lim_{\e\to 0}\min\left\{\frac{N}{2\pi\e}E^N_{\e}(v,Q_{\nu}):\,v(\e i)=u_{\nu}^{s,r}(\e i)\quad\forall \,  \e i\in\e\Z^d\text{ s.t.\ }\mathrm{dist}(\e i,\partial Q_{\nu})\leq 2\e\right\},
\end{equation}
cf. \cite[Remarks 5.9 \& 4.2(i)]{Bra-Cic-Ruf} for the fact that the width of the discrete boundary layer can be taken as $2\e$. In the above formula, $Q_{\nu}$ denotes a unit cube centered at the origin with two faces orthogonal to $\nu$ as in~\eqref{eq:defQ_nu}. The energy $E^N_{\e}(u,Q_{\nu})$ denotes the energy restricted to the set $Q_{\nu}$. More in general, for any non-empty set $A\subset\R^d$ and $u \colon \e\Z^d\to\S_N$ let us introduce for later purposes the localized functional
\begin{equation*}
E^N_{\e}(u,A)=\frac{1}{2}\sum_{\substack{\e i,\e j\in\e\Z^d\cap A\\ |i-j|=1}}\e^d|u(\e i)-u(\e j)|^2.
\end{equation*}

\section{Continuum limit for fixed $N$ as lattice spacing vanishes}
In this section we identify the variational limit of the $N$-clock model as $\e \to 0$ for the scaled energy $\tfrac{N}{2\pi\e}E^N_{\e}$. We start with the following auxiliary result that will be crucial to establish the lower bound.
\begin{lemma}\label{l.sinuslemma}
Let $k\in\N$. Then for all $\theta\in [0,\pi/k]$ it holds that
\begin{equation*}
\sin^2\Big(\frac{k\theta}{2}\Big)\geq k\sin^2\Big(\frac{\theta}{2}\Big) \,.
\end{equation*}
\end{lemma}	
\begin{proof}
We can assume that $k\geq 2$. Setting $y=\frac{k\theta}{2}$ we have that $y\in[0,\pi/2]$ and the claim reduces to
$\sin^2(y)\geq k\sin^2(y/k)$ for all $y\in[0,\pi/2]$. Since for $y\in[0,\pi/2]$ both $\sin(y)$ and $\sin(y/k)$ are non-negative, we can alternatively show that
\begin{equation}\label{eq:claimsin}
\sin(y)\geq\sqrt{k}\sin(y/k)\quad  \text{for all } y\in[0,\pi/2] \, .
\end{equation}
Let us define the auxiliary function $f_k(y)=\sin(y)-\sqrt{k}\sin(y/k)$. We show that it is strictly concave on $[0,\pi/2]$, so that its minimum is achieved at $y=0$ or $y=\pi/2$. Indeed, for $y\in[0,\pi/2]$ we have by the monotonicity of the sinus function that
\begin{equation*}
f_k''(y)=-\sin(y)+k^{-\tfrac{3}{2}}\sin(y/k)\leq-\sin(y)+k^{-\tfrac{3}{2}}\sin(y)\leq -\frac{1}{2}\sin(y),
\end{equation*}
so that $f_k''(y)<0$ whenever $y\in (0,\pi/2]$. Hence 
\begin{equation*}
\min_{y\in[0,\pi/2]}f_k(y)=\min\{f_k(0),f_k(\pi/2)\}=\min\{0,1-\sqrt{k}\sin(\pi/(2k))\}.
\end{equation*}
We conclude the proof once we show that $\sqrt{k}\sin(\pi/(2k)) \leq 1$ for all $k \geq 2$. Using that $\sin(x)<x$ for all $x>0$, for $k\geq 3$ we can bound the left hand side by
\begin{equation*}
\sqrt{k}\sin(\pi/(2k))\leq \frac{\pi}{2\sqrt{k}}\leq \frac{\pi}{2\sqrt{3}}<1,
\end{equation*}
while for $k=2$ we have $\sqrt{2}\sin(\pi/4)=1$. Thus $f_k(y)\geq 0$ for all $y\in[0,\pi/2]$ which yields \eqref{eq:claimsin} and concludes the proof.
\end{proof}
Next we establish a lower-semicontinuity result which helps to prove the lower bound.
\begin{lemma}\label{l.lsc}
For an open set $A\subset\Omega$ let $E(\, \cdot \, ,A) \colon L^1(A;\R^2)\to [0,+\infty]$ be the functional defined~by
\begin{equation*}
E(u,A)=\int_{A}|\nabla u|_{2,1}\,\mathrm{d}x+|\D^{(c)}u|_{2,1}(A)+\int_{J_u\cap A}\geo(u^-,u^+)|\nu_u|_1\,\mathrm{d}\mathcal{H}^{d-1}
\end{equation*} 
for $u\in BV(A;\SS^1)$ and extended to $+\infty$ otherwise. Then $u\mapsto E(u,A)$ is $L^1(A;\RR^2)$-lower semicontinuous.
\end{lemma}
\begin{proof}
	For an open set $I \subset \RR$ let $E^{1d}(\, \cdot \, , I ) \colon L^1(I;\RR^2) \to [0,+\infty]$ be defined by 
	\begin{equation*}
	E^{1d}(w,I) := \begin{cases}
	\displaystyle \int_{I}{|w'|}{\d t} + |\D^{(c)} w|(I) + \hspace{-0.3em}\sum_{t \in J_{w} \cap I}{\hspace{-0.3em} \geo\big(w^+(t),w^-(t)\big)} \, , & \!\!\! \text{ if } w \in  BV(I; \SS^1)  \, , \\
	+\infty \, , & \!\!\! \text{ otherwise.}
	\end{cases} 
	\end{equation*}
	By~\cite[Theorem 3.1]{ACL} (see also \cite[Remark 4.3]{ACL}), the functional $E^{1d}(\, \cdot \, , I)$ is the relaxation of 
	\begin{equation*}
	 \int_{I}{|w'|}{\d t} \, , \quad w \in W^{1,1}(I;\SS^1)
	\end{equation*}
	with respect to the strong topology of $L^1(I;\RR^2)$. In particular, it is lower semicontinuous.
	
	We next fix an open set $A \subset \Omega$ and $v_n, v \in L^1(A; \RR^2)$ such that $v_n \to v$ strongly in $L^1(A;\RR^2)$. We want to prove that 
	\begin{equation} \label{eq:H is lsc}
	E(v;A) \leq \liminf_{n \to +\infty} E(v_n; A) \, .
	\end{equation} 
	Without loss of generality, we assume that the right-hand side in~\eqref{eq:H is lsc} is finite and that the $\liminf$ is actually a limit. Since $|\D v_n|(A)\leq E(v_n;A)$ we obtain $v \in BV(A;\SS^1)$ and $v_n \wstar v$ weakly* in $BV(A;\RR^2)$. Note further that
	\begin{equation} \label{eq:H is a sum}
	E(v_n,A) = \sum_{\ell=1}^d \Big\{ \int_{A}{|\nabla v_n \, e_\ell|}{\d x} + |\D^{(c)} v_n \, e_\ell|(A) + \int_{J_{v_n} \cap A}{  \geo(v_n^+,v_n^-)|\nu_{v_n} \cdot e_\ell|}{\d \H^{d-1}} \Big\} \, .
	\end{equation}
		 
	Let us fix a direction $\xi \in \SS^1$, which plays the role of one of the coordinate directions~$e_\ell$. In the following we use the notation and the properties of slicing recalled in Subsection~\ref{subsec:BV}. We start by extracting a subsequence of $n$ (possibly depending on $\xi$ and which we do not relabel) such that the liminf
	\begin{equation*}
		\liminf_{n \to +\infty} \int_{A}{|\nabla v_n \, \xi|}{\d x} + |\D^{(c)} v_n \, \xi|(A) +  \int_{J_{v_n} \cap A}{  \geo(v_n^+,v_n^-)|\nu_{v_n} \cdot \xi|}{\d \H^{d-1}} 
	\end{equation*}
	is actually a limit. Moreover, since $v_n \to v$ strongly in $L^1(A;\RR^2)$, by Fubini's Theorem we extract a further subsequence (possibly depending on $\xi$ and which we do not relabel) such that 
	\begin{equation*}
	(v_n)^\xi_z \to v^\xi_z \quad \text{strongly in } L^1(A^\xi_z;\RR^2) \, , \quad \text{for } \H^{d-1}\text{-a.e.\ } z \in \Pi^\xi  .
	\end{equation*} 
	Moreover, we know that $v^\xi_z \in BV(A^\xi_z;\SS^1)$ for $\H^{d-1}$-a.e.\ $z \in \Pi^\xi$. 

	We observe now that the coarea formula (cf.\ \cite[formula (272)]{AmbFusPal} with $g=\geo(v_n^+,v_n^-)$, $E=J_{v_n}\cap A$, and $f$ the projection onto the orthogonal complement of $\xi$) implies 
	\begin{equation*}
	\int_{J_{v_n} \cap A}{\!\!\! \geo(v_n^+,v_n^-)|\nu_{v_n} \cdot \xi|}{\d \H^{d-1}} = \int \limits_{\Pi^\xi} \! \bigg[ \sum_{t \in J_{(v_n)^\xi_z} \cap A^\xi_z}{\hspace{-1.3em} \geo\Big(\big((v_n)^\xi_z\big)^+(t),\big((v_n)^\xi_z\big)^-(t)\Big)}  \bigg] \d \H^{d-1}(z) \, .
	\end{equation*}
	Hence, by the equality above and by Fatou's Lemma, we deduce that 
	\begin{equation} \label{eq:09041821}
	\begin{split}
	& \lim_{n \to +\infty} \int_{A}{|\nabla v_n \, \xi|}{\d x} + |\D^{(c)} v_n \, \xi|(A) +  \int_{J_{v_n} \cap A}{  \geo(v_n^+,v_n^-)|\nu_{v_n} \cdot \xi|}{\d \H^{d-1}}  \\
	& = \lim_{n \to +\infty} \int \limits_{\Pi^\xi} \! \bigg[ \int_{A^\xi_z}{\big|\big((v_n)^\xi_z\big)'\big|}{\d t} + |\D^{(c)} (v_n)^\xi_z|(A^\xi_z) + \hspace{-1.3em} \sum_{t \in J_{(v_n)^\xi_z} \cap A^\xi_z}{\hspace{-1.3em} \geo\Big(\big((v_n)^\xi_z\big)^+(t),\big((v_n)^\xi_z\big)^-(t)\Big)} \bigg] \d \H^{d-1}(z)\\
	& \geq \int \limits_{\Pi^\xi} \!  \liminf_{n \to +\infty} \bigg[ \int_{A^\xi_z}{\big|\big((v_n)^\xi_z\big)'\big|}{\d t} + |\D^{(c)} (v_n)^\xi_z|(A^\xi_z) + \hspace{-1.3em}\sum_{t \in J_{(v_n)^\xi_z} \cap A^\xi_z}{\hspace{-1.3em} \geo\Big(\big((v_n)^\xi_z\big)^+(t),\big((v_n)^\xi_z\big)^-(t)\Big)} \bigg] \d \H^{d-1}(z) \, .
	\end{split}
	\end{equation}
	From the one-dimensional lower semicontinuity result we infer that  
	\begin{equation*}
	\begin{split}
	& \liminf_{n\to+\infty} \int_{A^\xi_z}{\big|\big((v_n)^\xi_z\big)'\big|}{\d t} + |\D^{(c)} (v_n)^\xi_z|(A^\xi_z) + \hspace{-1.3em}\sum_{t \in J_{(v_n)^\xi_z} \cap A^\xi_z}{\hspace{-1.3em} \geo\Big(\big((v_n)^\xi_z\big)^+(t),\big((v_n)^\xi_z\big)^-(t)\Big)}\\
	& \quad = \liminf_{n \to +\infty} E^{1d}\big((v_n)^\xi_z, A^\xi_z\big)  \\
	& \quad \geq E^{1d}(v^\xi_z, A^\xi_z) = \int_{A^\xi_z}{\big|\big(v^\xi_z\big)'\big|}{\d t} + |\D^{(c)} v^\xi_z|(A^\xi_z) + \hspace{-1.3em}\sum_{t \in J_{v^\xi_z} \cap A^\xi_z}{\hspace{-1.3em} \geo\big((v^\xi_z)^+(t),(v^\xi_z)^-(t)\big)} 
	\end{split}
	\end{equation*}
	for $\H^{d-1}$-a.e.\ $z \in \Pi^\xi$. Integrating the inequality above with respect to $z \in \Pi^\xi$, again by the coarea formula, and by~\eqref{eq:09041821} we obtain that
	\begin{equation*}
	\begin{split}
	& \lim_{n \to +\infty} \int_{A}{|\nabla v_n \, \xi|}{\d x} + |\D^{(c)} v_n \, \xi|(A) + \!\!\! \int_{J_{v_n} \cap A}{\!\!\! \geo(v_n^+,v_n^-)|\nu_{v_n} \cdot \xi|}{\d \H^{d-1}}\\
	& \quad \geq \int_{A}{|\nabla v \, \xi|}{\d x} + |\D^{(c)} v \, \xi|(A) + \!\!\! \int_{J_{v } \cap A}{\!\!\! \geo(v^+,v^-)|\nu_{v} \cdot \xi|}{\d \H^{d-1}} .
	\end{split}
	\end{equation*}
	We conclude the proof of~\eqref{eq:H is lsc} by evaluating the last inequality for $\xi = e_1, \dots,e_d$, by~\eqref{eq:H is a sum}, and employing the superadditivity of the $\liminf$.
	\end{proof}
Now we can prove the lower bound for the $\Gamma$-limit of the functionals $\tfrac{N}{2\pi \e}E^N_{\e}$.

\begin{proposition}\label{p.gamma-liminf}
Let $u_{\e}\colon \e\Z^d\to\mathcal{S}_N$ and $u\in BV(\Omega;\mathcal{S}_N)$ be such that $u_{\e}\to u$ in $L^1(\Omega;\RR^2)$. Then
\begin{equation*}
\liminf_{\e\to 0}\frac{N}{2\pi\e}E_{\e}^N(u_{\e})\geq \frac{4 \sin^2\left(\tfrac{\theta_N}{2}\right)}{\theta_N^2}\int_{\Omega\cap J_u}\geo(u^-,u^+)|\nu_u|_1\,\mathrm{d}\mathcal{H}^{d-1}.
\end{equation*}
\end{proposition}
\begin{proof}
To simplify the notation we denote $\theta_N$ by $\theta$. Let $A\subset\subset \Omega$ be an open set. By~\eqref{eq:geo and eucl} it holds that 
\begin{equation*}
|u_\e(\e i) - u_\e(\e j)| = 2 \sin \big(\tfrac{1}{2} \geo(u_\e(\e i), u_\e(\e j) \big) \, .
\end{equation*}
Since $u_\e$ takes values in $\S_N$, the geodesic distance $\geo(u_\e(\e i),u_\e(\e j))$ is an integer multiple of $\theta$, i.e., there exists a $k \in \N$ (depending on $i$, $j$, and $\e$) such that $\geo(u_\e(\e i),u_\e(\e j)) = k \theta$. Note that $k \theta \leq \pi$. Hence from Lemma \ref{l.sinuslemma} we infer that
\begin{equation*}
	\begin{split}
		\frac{1}{2}|u_\e(\e i) - u_\e(\e j)|^2 &= 2 \sin^2 \big(\tfrac{1}{2} \geo(u_\e(\e i), u_\e(\e j) \big)  = 2 \sin^2\Big(\frac{k \theta}{2}\Big) \geq 2 k \sin^2\Big(\frac{ \theta}{2}\Big) \\
		&= 2\geo(u_\e(\e i), u_\e(\e j))\frac{\sin^2(\tfrac{\theta}{2} )}{\theta} \, .
	\end{split}
\end{equation*}
Since $u_{\e}$ is piecewise constant on cubes of the form $Q=(-\e/2,\e/2)^d+z$ with $z\in\Z^d$, we obtain that for $\e$ small enough
\begin{equation*}
\frac{N}{2\pi\e}E_{\e}^N(u_{\e})\geq \frac{4 \sin^2\left(\frac{\theta}{2}\right)}{\theta^2}\int_{A\cap J_u}\geo(u_{\e}^-,u_{\e}^+)|\nu_{u_{\e}}|_1\,\mathrm{d}\mathcal{H}^{d-1},
\end{equation*}
where we also used that $N=2\pi/\theta$ and that the discrete energy counts each interaction twice. Note that by Lemma \ref{l.lsc} the functional
\begin{equation*}
u\mapsto \int_{A\cap J_u}\geo(u^-,u^+)|\nu_u|\,\mathrm{d}\mathcal{H}^{d-1}
\end{equation*}
is $L^1(A;\RR^2)$-lower semicontinuous on $BV(A;\S_N)$, as it is the restriction of a lower semicontinuous functional to a closed subset of $BV(A;\S_N)$. Thus letting $\e\to 0$ we deduce that
\begin{equation*}
\liminf_{\e\to 0}\frac{N}{2\pi\e}E_{\e}^N(u_{\e})\geq \frac{4 \sin^2\left(\tfrac{\theta}{2}\right)}{\theta^2}\int_{A\cap J_u}\geo(u^-,u^+)|\nu_u|_1\,\mathrm{d}\mathcal{H}^{d-1}.
\end{equation*}
The claim now follows from the arbitrariness of $A\subset\subset\Omega$.
\end{proof}

We next prove that the corresponding upper bound for the $\Gamma$-limit.
\begin{proposition}\label{p.Gamma-limsup}
Let $u\in BV(\Omega;\S_N)$. Then there exists a sequence $u_{\e} \colon \e\Z^d\to\S_N$ such that $u_{\e}\to u$ in $L^1(\Omega;\RR^2)$ and
\begin{equation*}
\limsup_{\e\to 0}\frac{N}{2\pi\e}E_{\e}^N(u_{\e})=\frac{4 \sin^2\left(\tfrac{\theta_N}{2}\right)}{\theta_N^2}\int_{\Omega\cap J_u}\geo(u^-,u^+)|\nu_u|_1\,\mathrm{d}\mathcal{H}^{d-1}.
\end{equation*}	
\end{proposition}
\begin{proof}
	To simplify the notation we denote $\theta_N$ by $\theta$. Due to the discussion in Section~\ref{sec:known}, the $\Gamma$-limit of $\frac{N}{2\pi\e}E_{\e}^N$ has the form~\eqref{eq:integral representation}. To prove the upper bound it suffices to define a suitable candidate for the minimum problem \eqref{eq:cellformula} whose energy can be bounded in the limit as $\e\to 0$ by $4\sin^2(\tfrac{\theta}{2})\theta^{-2}\geo(s,r)|\nu|_1$. Write $s=\exp(\iota k_s\theta)$ and $r=\exp(\iota k_r\theta)$ with $0\leq k_s,k_r\leq N-1$. We will treat the case when $k_r=0$, i.e. $r=(1,0)$, and $0<k_s\theta\leq\pi$. The construction we provide can then be composed with a rotation in the co-domain to cover the general case. The idea is to define a candidate whose angular variable jumps by $\theta$ along the discretization of~$k_s$ parallel hyperplanes orthogonal to $\nu$, where all hyperplanes are $\mathcal{O}(\e)$-close to the hyperplane $\Pi_{\nu}:=\{x\in\R^d  :  \, x \cdot \nu =0\}$. The correction in order to satisfy the boundary condition will be of lower order. In formulas, let $u_{\e} \colon \e\Z^d\to\S_N$ be defined by
\begin{equation*}
u_{\e}(\e i):=
\begin{cases}
\exp\big(\iota \min\big\{k_s,\max\{0,\lfloor i \cdot \nu \rfloor\}\big\}\theta\big) &\mbox{if } \mathrm{dist}(\e i,\partial Q_{\nu})>2\e \, ,
\\
u^{s,r}_{\nu}(\e i) &\mbox{if } \mathrm{dist}(\e i,\partial Q_{\nu})\leq 2\e \, ,
\end{cases}
\end{equation*}
where $\lfloor x\rfloor$ denotes the integer part of $x$. Hence for all $\e i\in\e\Z^d\cap Q_{\nu}$ such that $\e i \cdot \nu \leq 0$ we have $u_{\e}(\e i)=r$, while for all $\e i\in\e\Z^d$ with $\e i \cdot \nu \geq k_s\e$ we have $u_{\e}(\e i)=s$, so that for non-vanishing interactions at least one point belongs to the set
\begin{equation*}
H^{k_s}_{\e}:=\{x\in Q_{\nu}:\,  x \cdot \nu \in (0,\e k_s)\} \, .
\end{equation*}
Note that we have the volume bound
\begin{equation*}
|H^{k_s}_{2\e}\cap \{\text{dist}(x,\partial Q_{\nu})\leq 4\e\}|\leq Ck_s\e^2,
\end{equation*}
where $C$ depends only on the dimension. Hence, for $\e$ small enough,
\begin{equation}\label{eq:cardinality_closetobdr}
\#\{z\in\Z^d:\e z\in H^{k_s}_{2\e}\cap \{\text{dist}(x,\partial Q_{\nu})\leq 3\e\}\}\leq C k_s\e^{2-d}
\end{equation}
To simplify notation, we also define the auxiliary function $v_{\e} \colon \e\Z^d\to\S_N$ by 
\begin{equation*}
v_{\e}(\e i):=\exp\big(\iota \min\big\{k_s,\max\{0,\lfloor i \cdot \nu \rfloor\}\big\}\theta\big) \, .
\end{equation*}
Since $|u_{\e}(\e i)-u_{\e}(\e j)|^2\leq 4$ it follows from the almost additivity of the set function $A\mapsto E_{\e}^N(u,A)$ that the energy of $u_{\e}$ can be estimated by
\begin{align*}
\frac{N}{2\pi\e}E_{\e}^N(u_{\e},Q_{\nu})&\leq \frac{N}{2\pi\e}E_{\e}^N(u_{\e},H^{k_s}_{2\e}\cap \{\text{dist}(x,\partial Q_{\nu})\leq 3\e\})+\frac{N}{2\pi\e}E_{\e}^N(v_{\e},Q_{\nu})
\\
&\leq CNk_s\e+\frac{N}{2\pi\e}E_{\e}^N(v_{\e},Q_{\nu})\leq CN^2 \e+\frac{N}{2\pi\e}E_{\e}^N(v_{\e},Q_{\nu}) \, .
\end{align*}
As $N$ is fixed, the first term in the right hand side vanishes when $\e\to 0$. Since $u_{\e}$ is admissible for the minimum problem \eqref{eq:cellformula} it suffices to show that
\begin{equation}\label{eq:claim1}
\limsup_{\e\to 0}\frac{N}{2\pi\e}E_{\e}^N(v_{\e},Q_{\nu})\leq \frac{4\sin^2(\tfrac{\theta}{2} )}{\theta^2}k_s\theta|\nu|_1=\frac{4\sin^2(\tfrac{\theta}{2})}{\theta}k_s|\nu|_1 \, .
\end{equation}

We start by noticing that when $\e i,\e j\in\e\Z^d\cap Q_{\nu}$ are such that $|i-j|=1$ and $v_{\e}(\e i)\neq v_{\e}(\e j)$, then $\e i \cdot \nu \neq \e j \cdot \nu$. Without loss of generality, we assume $\e i \cdot \nu > \e j \cdot \nu$. Note that $j \cdot \nu \geq 0$. Indeed, if instead $j \cdot \nu <0$, then $i \cdot \nu  <1$ and thus $v_{\e}(\e i)=v_{\e}(\e j)$, which  contradicts $v_{\e}(\e i)\neq v_{\e}(\e j)$. Moreover, by a similar argument we also know that $k_s+1 > i \cdot \nu$. To sum up, we have that 
\begin{equation}
	0 \leq \e j \cdot \nu < \e i \cdot \nu < (k_s+1)\e \, .
\end{equation}
Finally, we have the estimate $|(\e i-\e j) \cdot \nu |\leq \e$, so that by \eqref{eq:geo and eucl} 
\begin{equation}\label{eq:thetajump}
|v_{\e}(\e i)-v_{\e}(\e j)|^2=4\sin^2(\tfrac{\theta}{2}) \, .
\end{equation} 

It remains to count the interactions. We will first split them according to their jump between~$\e j \cdot \nu$ and $ \e i \cdot \nu$. More precisely, for a natural number $k\in \{1,\ldots,k_s\}$ we set
\begin{equation*}
I_{k,\e}:=\{(\e i,\e j)\in(\e\Z^d\cap Q_{\nu})^2:\ |i-j|=1 \, ,\ \lfloor  j \cdot \nu \rfloor=k-1 \, , \ \lfloor  i \cdot \nu \rfloor=k\}
\end{equation*} 
Note that a pair $(\e i, \e j) \in I_{k,\e}$ is only counted once. Since each pair of interactions in the energy is counted twice, we deduce from \eqref{eq:thetajump} and the equality $N/2\pi=1/\theta$ that
\begin{equation*}
\frac{N}{2\pi\e}E_{\e}^N(v_{\e},Q_{\nu})\leq \frac{4\sin^2(\tfrac{\theta}{2})}{\theta}\sum_{k=1}^{k_s}\e^{d-1}\# I_{k,\e} \, .
\end{equation*}
We deduce then~\eqref{eq:claim1} from the asymptotic formula 
\begin{equation}
	\limsup_{\e\to 0}\e^{d-1} \# I_{k,\e}\leq|\nu|_1 \, .
\end{equation}

The above formula can be justified as follows: first further subdivide the set $I_{k,\e}$ into the $d$ disjoint sets $(I_{k,\e}^\ell)_{\ell=1}^d$ defined~by
\begin{equation*}
I_{k,\e}^\ell:=\{(\e i,\e j)\in I_{k,\e}:\,(i-j) \text{ is parallel to } e_\ell\} \quad \text{for } \ell = 1, \dots, d \,.
\end{equation*}
Observe that $I_{k,\e}=\bigcup_{\ell=1}^d I^\ell_{k,\e}$ and that if there exists a pair $(\e i,\e j) \in I_{k,\e}^\ell$, then $\nu_\ell\neq 0$. Indeed, in that case the hyperplane $H_{\nu}=\{ x \cdot \nu =0\}$ does not contain $(i-j)$, and in turn $e_\ell$,  by definition of~$I_{k,\e}$. Next we estimate where the line $\e j+\R e_\ell$ intersects the hyperplane $H_{\nu}=\{ x \cdot \nu =0\}$. It does in a unique point $\e j+ \lambda e_\ell$ when $I_{k,\e}^\ell\neq\emptyset$. Since $0\leq  \e j \cdot \nu \leq k\e$ it follows that
\begin{equation*}
|\lambda|\leq \frac{k\e}{|\nu_\ell|} \, .
\end{equation*} 
Therefore, given $t>1$, for $\e=\e(t)$ small enough the intersection point is contained in $tQ_{\nu} \cap H_{\nu}$. Since by definition the mapping $I_{k,\e}^\ell\ni (\e i,\e j) \mapsto \e j-(\e j \cdot e_\ell) e_\ell$ is injective, we obtain that
\begin{equation*}
\# I_{k,\e}^\ell\leq\#\{\e i\in\e\Z^d:\,\e i\in \Pi_{x_\ell=0}(tQ_{\nu}\cap H_{\nu})\} \, ,
\end{equation*} 
where $\Pi_{x_\ell=0}$ denotes the projection onto the subspace $\{x_\ell=0\}$. In particular, it holds that
\begin{equation*}
\e^{d-1}\#I_{k,\e}= \sum_{\ell=1}^d\e^{d-1}\#I_{k,\e}^\ell\leq\sum_{\ell=1}^d\e^{d-1}\#\big(\e\Z^d\cap\Pi_{x_\ell=0}(tQ_{\nu}\cap H_{\nu})\big) \, .
\end{equation*}
\begin{figure}[H]
	\includegraphics{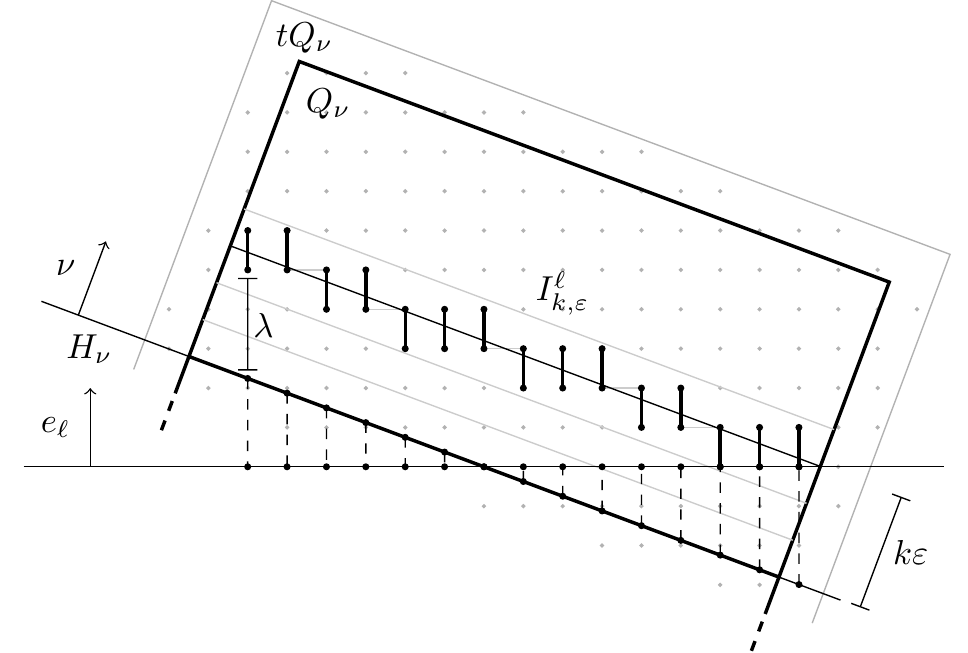}
\caption{Counting the number of points in $I^\ell_{k,\e}$.}

\label{fig:counting}
\end{figure}
\noindent By elementary geometric considerations we can bound the cardinality via a $(d-1)$-dimensional volume as
\begin{equation*}
	\lim_{\e\to 0}\e^{d-1}\left(\#\e\Z^d\cap\Pi_{x_\ell=0}(tQ_{\nu}\cap H_{\nu})\right)=\mathcal{H}^{d-1}(\Pi_{x_\ell=0}(tQ_{\nu}\cap H_{\nu}))=t^{d-1}\mathcal{H}^{d-1}(\Pi_{x_\ell=0}(Q_{\nu}\cap H_{\nu})) \, .
\end{equation*}
Since $t > 1$ was arbitrary we deduce that
\begin{equation*}
	\limsup_{\e\to 0}\e^{d-1} \# I_{k,\e}^\ell\leq \mathcal{H}^{d-1}(\Pi_{x_\ell=0}(Q_{\nu}\cap H_{\nu})) \, .
\end{equation*}
We claim that the right hand side term equals $|\nu_\ell|$, which then concludes the proof summing over~$\ell$. This is a consequence of the coarea formula in the form \cite[Theorem 2.93]{AmbFusPal} taking $f$ to be the projection $\Pi_{x_\ell=0}$ and $E=Q_{\nu}\cap H_{\nu}$ and using the fact that the $(d-1)$-dimensional coarea factor of the projection $\Pi_{x_\ell=0}$ on the tangent space $H_{\nu}$ is given by $|\nu_\ell|$ (cf. \cite[formula (3.110)]{AmbFusPal}).
\end{proof}

\section{Limit of the continuum functional for large $N$}
In this section we study the $\Gamma$-convergence of the limit functionals $E_N$ defined on $L^1(\Omega;\RR^2)$ by 
\begin{equation} \label{eq:def of EN}
E_N(u) := \begin{cases}
	\displaystyle \frac{4 \sin^2\big(\tfrac{\theta_N}{2}\big)}{\theta_N^2}\int_{\Omega\cap J_u}\geo(u^-,u^+)|\nu_u|_1\,\mathrm{d}\mathcal{H}^{d-1} & \text{if } u \in BV(\Omega; \S_N) \, , \\
+ \infty & \text{otherwise,}
	\end{cases}
\end{equation}
as $N\to +\infty$,  where we write $\theta_N$ to stress the dependence  on $N$ of the minimal angle between vectors in $\S_N$. We show that the $\Gamma$-limit of $E_N$ coincides with the functional derived in \cite{COR2} in the regime $N=N_{\e}\ll \frac{1}{\e|\log\e|}$ and $d=2$. More precisely, we define the functional
\begin{equation} \label{eq:def of E}
	\begin{split}
		E(u) := \begin{cases}	
			\displaystyle \int_{\Omega}|\nabla u|_{2,1}\,\mathrm{d}x+|\D^{(c)}u|_{2,1}(\Omega)+\int_{\Omega\cap J_u}\!\!\! \geo(u^-,u^+)|\nu_u|_1\,\mathrm{d}\mathcal{H}^{d-1},   & \text{if } u \in BV(\Omega;\mathbb{S}^1) \, , \\
			+ \infty & \text{otherwise,}
		\end{cases} 
	\end{split}
\end{equation}
for $u \in L^1(\Omega;\RR^2)$.

We first state and proof the lower bound together with a compactness result.
\begin{proposition}[Lower bound and compactness]\label{prop.lb+compact_cont}
Let $u_N\in BV(\Omega;\S_N)$ be a sequence such that
\begin{equation*}
\sup_N E_N(u_N)<+\infty \, .
\end{equation*}
Then up to subsequences $u_N\to u\in BV(\Omega;\mathbb{S}^1)$ strongly in $L^1(\Omega;\R^2)$. Moreover, for any sequence $u_N\in BV(\Omega;\S_N)$ and $u\in BV(\Omega;\mathbb{S}^1)$ such that $u_N\to u$ in $L^1(\Omega;\R^2)$ it holds that
\begin{equation*}
\liminf_{N\to +\infty}E_N(u_N)\geq E(u) \, .
\end{equation*}
\end{proposition}
\begin{proof}
Since $\geo(u,v)\geq |u-v|$, the functionals $E_N$ satisfy
\begin{equation*}
E_N(u)\geq\frac{4\sin^2(\tfrac{\theta_N}{2})}{\theta_N^2}|\D u|(\Omega) \, .
\end{equation*}
Note that $\theta_N=2\pi/N$ implies $\theta_N\to 0$ as $N\to +\infty$. Hence
\begin{equation}\label{eq:limittheta}
\lim_{N\to +\infty}\frac{4\sin^2(\tfrac{\theta_N}{2})}{\theta_N^2}=1 \, .
\end{equation}
Thus the compactness statement follows from the inclusion $\S_N\subset\mathbb{S}^1$ and standard compactness results in $BV(\Omega;\R^2)$.

In order to prove the lower bound, note that
\begin{equation*}
	\begin{split}
		E_N(u) & \geq \frac{4\sin^2(\tfrac{\theta_N}{2})}{\theta_N^2}\left(\int_{\Omega}|\nabla u|_{2,1}\,\mathrm{d}x+|\D^{(c)}u|_{2,1}(\Omega)+\int_{\Omega\cap J_{u}}\geo(u^-,u^+)|\nu_{u}|_1\,\mathrm{d}\mathcal{H}^{d-1}\right) \\
		& = \frac{4\sin^2(\tfrac{\theta_N}{2})}{\theta_N^2} E(u) 
	\end{split}
\end{equation*}
for all $u\in BV(\Omega;\mathbb{S}^1)$, cf.~\eqref{eq:def of EN}--\eqref{eq:def of E}. The functional $E$ is $L^1(\Omega;\RR^2)$-lower semicontinuous by Lemma~\ref{l.lsc}. Hence, the claim follows from \eqref{eq:limittheta}.
\end{proof}

We now establish the upper bound via several approximations combined with a relaxation result for integral functionals defined on $W^{1,1}(\Omega;\mathbb{S}^1)$.

 We recall here the density result proven in~\cite{Bet}.  Let 
\begin{equation*}
	\begin{split}
		\mathcal{R}^\infty_1(\Omega;\mathbb{S}^1) := \{ u \in W^{1,1}(\Omega; \mathbb{S}^1) : \, & u \in C^\infty(\Omega \sm \Sigma;\mathbb{S}^1 ) \, ,\  \Sigma = \mathop{\textstyle \bigcup_{i=1}^m} \Sigma_h \, , \ m \in \mathbb{N}  \\ 
		& \Sigma_h \text{ closed  subset of a $(d-2)$-dimensional manifold}    \} \, .
	\end{split}
\end{equation*}

\begin{theorem} \label{thm:Bethuel}
	The class $\mathcal{R}^\infty_1(\Omega;\mathbb{S}^1)$ is dense in $W^{1,1}(\Omega;\mathbb{S}^1)$ with respect to the strong convergence in $W^{1,1}(\Omega;\RR^2)$.
\end{theorem}

\begin{proposition}[Upper bound]\label{prop.upperbound_cont}
For every function $u\in BV(\Omega;\mathbb{S}^1)$ there exists a sequence $u_N\in BV(\Omega;\S_N)$ such that $u_N\to u$ strongly in $L^1(\Omega;\R^2)$ and 
\begin{equation*}
\lim_{N\to +\infty} E_N(u_N)= E(u) \, .
\end{equation*} 
\end{proposition}
\begin{proof}

	Thanks to Proposition~\ref{p.gamma-liminf}, it is enough to prove that for every $u \in BV(\Omega;\SS^1)$ there exists a sequence $u_N\in BV(\Omega;\S_N)$ such that $u_N\to u$ strongly in $L^1(\Omega;\R^2)$ and 
	\begin{equation} \label{eq:upper bound}
	\limsup_{N\to +\infty} E_N(u_N) \leq E(u) \, .
	\end{equation} 

\step{1}\ (Reducing to the case $u \in W^{1,1}(\Omega;\mathbb{S}^1))$. 
Let us start by considering the functional given by
\begin{equation} \label{eq:functional in W11}
	\int_\Omega |\nabla u|_{2,1}  \d x \, , \quad \text{if } u \in W^{1,1}(\Omega;\SS^1) 
\end{equation}
and by $+\infty$ otherwise in $L^1(\Omega; \RR^2)$.  This functional satisfies all the assumptions of the functionals studied in~\cite{ACL}, cf.\ assumptions (H1)--(H5) therein. Then, by~\cite[Theorem~3.1]{ACL}, its relaxation is given by 
\begin{equation*}
	\int_\Omega |\nabla u|_{2,1} \d x  + |\D^{(c)}u|_{2,1}(\Omega) + \int_{\Omega \cap J_u} K(u^-,u^+,\nu_u) \d \H^{d-1} , \quad \text{if } u \in BV(\Omega;\SS^1)
\end{equation*}
and by $+ \infty$ otherwise in $L^1(\Omega;\RR^2)$. The density of the surface energy $K \colon \SS^1 \x \SS^1 \x \SS^{d-1} \to [0,+\infty)$ is characterized by the formula
\begin{equation*}
K(a,b,\nu) := \inf \Big\{ \int_{Q_\nu} |\nabla \psi|_{2,1} \d x : \, \psi \in \mathcal{P}(a,b,\nu) \Big\} \, ,
\end{equation*}
where $Q_\nu$ is a unit cube centered at the origin with two faces orthogonal to $\nu$ and $\mathcal{P}(a,b,\nu)$ is the collection of all $\psi \in W^{1,1}(Q_\nu;\SS^1)$ with $\psi(x) = a$ if $x \cdot \nu = -\frac{1}{2}$, $\psi(x) = b$ if $x \cdot \nu = \frac{1}{2}$, and~$\psi$ is periodic with period $1$ in the direction orthogonal to $\nu$. In particular, $\mathcal{P}(a,b,\nu)$ contains the collection of functions with a one-dimensional profile in the direction $\nu$, i.e., functions $\psi \in W^{1,1}(Q_\nu;\SS^1)$ such that there exists a curve $\gamma \in W^{1,1}((-\frac{1}{2}, \frac{1}{2}); \SS^1)$ with $\gamma(-\frac{1}{2}) = a$, $\gamma(\frac{1}{2}) = b$ satisfying $\psi(x) = \gamma(x \cdot \nu)$. For such functions we have $\nabla \psi(x) = \gamma'(x \cdot \nu) \otimes \nu$ and therefore, since $|\gamma'(x \cdot \nu) \otimes \nu |_{2,1} = |\gamma'(x \cdot \nu)| \, |\nu|_1$, 
\begin{equation*}
K(a,b,\nu) \leq \int_{Q_\nu} |\nabla \psi|_{2,1} \d x  =  |\nu|_1 \int_{Q_\nu} |\gamma'(x \cdot \nu)| \d x = |\nu|_1  \int_{-\frac{1}{2}}^\frac{1}{2} |\gamma'(t)| \, \d t \, .
\end{equation*}
Taking the infimum over all such curves $\gamma \in W^{1,1}((-\frac{1}{2}, \frac{1}{2}); \SS^1)$ with $\gamma(-\frac{1}{2}) = a$, $\gamma(\frac{1}{2}) = b$, we conclude that 
\begin{equation*}
K(a,b,\nu) \leq \geo(a,b) |\nu|_1 \, .
\end{equation*}
In particular, the relaxation of~\eqref{eq:functional in W11} is smaller than $E$, cf.~\eqref{eq:def of E}. This entails that  for every $u \in BV(\Omega;\SS^1)$ there exists a sequence $u_j \in W^{1,1}(\Omega;\SS^1)$ such that $u_j \to u$ in $L^1(\Omega;\RR^2)$ and 
\begin{equation*}
	\limsup_{j \to +\infty} \int_\Omega |\nabla u_j|_{2,1}  \d x \leq E(u) \, .
\end{equation*}
Thanks to this property and to a diagonal argument, it is enough to prove the upper bound~\eqref{eq:upper bound} assuming $u \in W^{1,1}(\Omega;\SS^1)$. 

\step{2}\ (Extending outside $\Omega$). Let $u \in W^{1,1}(\Omega;\SS^1)$. There exists $t>0$ and a bi-Lipschitz map $\Gamma\colon (\partial \Omega \x(-t,t))\to \Gamma(\partial \Omega\x (-t,t))$ such that $\Gamma(x,0)=x$ for all $x\in\partial \Omega$, $\Gamma(\partial \Omega \x (-t,t))$ is an open neighborhood of $\partial \Omega$ and
\begin{equation}\label{eq:bicollar}
\Gamma(\partial \Omega\x (-t,0))\subset \Omega,\quad\quad \Gamma(\partial \Omega \times(0,t))\subset\RR^2\setminus\ol \Omega\,.
\end{equation}
This result is a consequence of \cite[Theorem 7.4 \& Corollary 7.5]{Luu-Vae}; details can be found for instance in \cite[Theorem 2.3]{Lic}. The extension of $u$ is then achieved via reflection. More precisely, for a sufficiently small $\tilde t>0$ we define it on $\tilde \Omega$ with $\tilde \Omega=\Omega+B_{\tilde t}(0)$ by
\begin{equation}\label{eq:extensionbyreflection}
\tilde{u}(x)=
\begin{cases}
u(\Gamma(P(\Gamma^{-1}(x)))) &\mbox{if $x\notin \Omega$} \, ,
\\
u(x) &\mbox{otherwise,}
\end{cases}
\end{equation}
where $P(x,\tau)=(x,-\tau)$. Since $\Gamma$ is bi-Lipschitz, we have that $\tilde{u}\in W^{1,1}(\tilde \Omega;\SS^1)$ and by a change of variables we can bound the $L^1$-norm of its gradient via
\begin{equation}\label{eq:gradbound}
\int_{\tilde \Omega}{|\nabla \tilde{u}|}{\d x}\leq \int_{\Omega}{|\nabla u|}{\d x}+C_{\Gamma} \int_{\tilde \Omega \setminus \Omega}{|(\nabla u)\circ \Gamma\circ P\circ\Gamma^{-1}|}{\d x}\leq C_{\Gamma}\int_{\Omega}{|\nabla u|}{\d x} \, ,
\end{equation}
where the constant $C_{\Gamma}$ depends only on the bi-Lipschitz properties of $\Gamma$ and the dimension. With an abuse of notation we will denote the extended function $\tilde u \in W^{1,1}(\tilde \Omega;\SS^1)$ again by~$u$. 

\step{3}\ (Reducing to the case $u \in \mathcal{R}^\infty_1(\tilde \Omega;\SS^1)$). Given $u \in W^{1,1}(\Omega;\SS^1)$, we extend it to a function in $W^{1,1}(\tilde \Omega;\SS^1)$ as in the previous step. By Theorem~\ref{thm:Bethuel} there exists a sequence $u_j \in \mathcal{R}^\infty_1(\tilde \Omega;\SS^1)$  such that $u_j \to u$ strongly in $W^{1,1}(\tilde \Omega;\RR^2)$. In particular, 
\begin{equation*}
	\lim_{j \to +\infty} \int_{\Omega} |\nabla u_j|_{2,1}  \d x = \int_{\Omega} |\nabla u|_{2,1} \, .
\end{equation*}
Hence, by a diagonal argument it is enough to prove the upper bound~\eqref{eq:upper bound} assuming $u \in \mathcal{R}^\infty_1(\tilde \Omega;\SS^1)$. 

\step{4}\ (Reducing to the case of piecewise constant $\SS^1$-valued maps). Let $u \in \mathcal{R}^\infty_1(\tilde \Omega;\SS^1)$. Then there exists $\Sigma = \bigcup_{h=1}^m \Sigma_h$ with $\Sigma_h$  closed  subset of a smooth $(d-2)$-dimensional manifold such that $u \in C^\infty(\tilde \Omega \sm \Sigma;\SS^1) \cap W^{1,1}(\tilde \Omega;\SS^1)$. We construct now an approximation of~$u$ through $\SS^1$-valued maps which are piecewise constant on a lattice of spacing $\lambda > 0$. Let us consider the family of half-open cubes 
\begin{equation*}
	I_\lambda(\lambda z) = \lambda z + \lambda [0,1)^d \, , \quad z \in \ZZ^d 
\end{equation*}
and the set 
\begin{equation*}
	\Omega^{\lambda} := \bigcup \{  I_\lambda(\lambda z) : \, z \in \ZZ^d \text{ such that } I_\lambda(\lambda z) \cap \Omega \neq \emptyset \} \, .
\end{equation*}
Let $\Omega'$ be such that $\Omega \subset \subset \Omega' \subset \subset \tilde \Omega$. For $\lambda$ small enough we have $\Omega^\lambda \subset \subset \Omega' \subset \subset \tilde \Omega$. We now define the piecewise constant function $u_\lambda \colon \Omega^\lambda \to \SS^1$ as follows. Let $z \in \ZZ^d$ be such that $I_\lambda(\lambda z) \subset \Omega^\lambda$. If $\ol{I_\lambda(\lambda z)} \cap \Sigma = \emptyset$, the map $u$ is $C^\infty$ in the interior of $I_\lambda(\lambda z)$ and thus it admits a lifting $\varphi_z$  (unique up to a multiple integer of $2 \pi$), which is $C^\infty$ in the interior of $I_\lambda(\lambda z)$, namely $u = \exp(\iota \varphi_z)$ in $I_\lambda(\lambda z)$. We consider the average
\begin{equation*}
	\ol \varphi_z := \frac{1}{\lambda^d} \int_{I_\lambda(\lambda z)} \varphi_z(x) \d x 
\end{equation*}
and we set $u_\lambda(x) := \exp(\iota \ol \varphi_z)$ for $x \in I_\lambda(\lambda z)$. If, instead, $\ol{I_\lambda(\lambda z)} \cap \Sigma \neq \emptyset$ we put $u_\lambda(x) := e_1$ for $x \in I_\lambda(\lambda z)$ (the precise value $e_1$ being not relevant). 

We remark that $u_\lambda \to u$ strongly in $L^1(\Omega;\RR^2)$.  Indeed, let $B$ be a ball such that $B \subset \subset \Omega \sm \Sigma$. Since $B$ is simply connected and $u \in C^\infty(B;\SS^1)$, there exists a lifting $\varphi \in C^\infty(B;\RR)$, namely, $u = \exp(\iota \varphi)$ in $B$. If $I_\lambda(\lambda z) \cap B \neq \emptyset$, then $\ol{I_\lambda(\lambda z)} \cap \Sigma = \emptyset$ for $\lambda$ small enough. In particular, we can consider the lifting $\varphi_z$ of $u$ in~$I_{\lambda}(\lambda z)$ used in the definition of~$u_\lambda$. By uniqueness of the liftings up to integer multiples of $2\pi$, there exists a $k_z \in \ZZ$ such that $\varphi_z = \varphi + 2 \pi k_z$. This entails 
\begin{equation*}
	\ol \varphi_z = \frac{1}{\lambda^d} \int_{I_\lambda(\lambda z)} \varphi (y) \d y  + 2 \pi k_z \, .
\end{equation*}
Given $x\in B$, we consider a family of cubes $I_\lambda(\lambda z_\lambda) \ni x$. By Lebesgue's differentiation theorem 
\begin{equation*}
	\frac{1}{\lambda^d} \int_{I_\lambda(\lambda z_\lambda)} \varphi (y) \d y \to \varphi(x)
\end{equation*}
for $\mathcal{L}^d$-a.e.\ $x \in B$. Then $u_\lambda \to u$ a.e.\ in $\Omega$ and by dominated convergence we obtain $u_\lambda \to u$ in $L^1(\Omega;\RR^2)$.

Let us prove that 
\begin{equation} \label{eq:limsup on W11}
	\limsup_{\lambda \to 0} \int_{\Omega^\lambda \cap J_{u_\lambda}} \geo(u_\lambda^-,u_\lambda^+)|\nu_{u_\lambda}|_1\,\mathrm{d}\mathcal{H}^{d-1} \leq \int_\Omega |\nabla u|_{2,1} \d x \, .
\end{equation}
For $i \in \{1,\dots, d\}$ we define the families of indices 
	\begin{equation*}
		\begin{split}
			\mathcal{Z}_i(\lambda) & := \{ z \in \ZZ^d : \,   I_{\lambda}(\lambda z)  \cup  I_{\lambda}(\lambda (z+e_i))  \subset \Omega^\lambda  \} \, , \\
			\mathcal{G}_i(\lambda) & := \{ z \in \mathcal{Z}_i(\lambda)  : \,  \ol{I_{\lambda}(\lambda z)} \cap \Sigma = \emptyset \quad \text{and} \quad \ol{I_{\lambda}(\lambda (z+e_i))} \cap \Sigma = \emptyset\} \, , \\
			\mathcal{B}_i(\lambda) & := \{ z \in \mathcal{Z}_i(\lambda) : \,  \ol{I_{\lambda}(\lambda z)} \cap \Sigma \neq \emptyset \quad \text{or} \quad  \ol{I_{\lambda}(\lambda (z+e_i))}  \cap \Sigma \neq \emptyset \} \, .
		\end{split}
	\end{equation*} 

	Let $z \in \mathcal{G}_i(\lambda)$. As in the definition of $u_\lambda$, we let $\varphi_z$ and $\varphi_{z+e_i}$ be the liftings of $u$ in $I_{\lambda}(\lambda z)$ and $I_{\lambda}(\lambda (z+e_i))$, respectively. Moreover, since $u$ is $C^\infty$ in the interior of the rectangle $I_{\lambda}(\lambda z) \cup I_{\lambda}(\lambda (z+e_i))$, it admits a $C^\infty$ lifting $\varphi$ such that $u = \exp(\iota \varphi)$ in $I_{\lambda}(\lambda z) \cup I_{\lambda}(\lambda (z+e_i))$. By uniqueness of the liftings up to integer multiples of $2\pi$, there exist $k_z, k_{z+e_i} \in \ZZ$ such that $\varphi_z = \varphi + 2 \pi k_z$ in $I_{\lambda}(\lambda z)$ and $\varphi_{z+e_i} = \varphi + 2 \pi k_{z+e_i}$ in $I_{\lambda}(\lambda (z+e_i))$. Note that 
\begin{equation*}
	\ol \varphi_z = \frac{1}{\lambda^d} \int_{I_\lambda(\lambda z)} \varphi (x) \d x  + 2 \pi k_z \, , \quad \ol \varphi_{z+e_i} = \frac{1}{\lambda^d} \int_{I_\lambda(\lambda (z+e_i))} \varphi (x) \d x  + 2 \pi k_{z+e_i} \, .
\end{equation*}
Now we are in a position to estimate
\begin{equation} \label{eq:estimate on good cubes}
	\begin{split}
		\geo\big(u_\lambda(\lambda(z+e_i)), u_\lambda(\lambda z)\big) & = \geo\big(\exp(\iota \ol \varphi_{z+e_i}), \exp(\iota \ol \varphi_z)\big) \\
		& \leq \frac{1}{\lambda^d} \Big| \int_{I_\lambda(\lambda (z+ e_i))} \varphi (x) \d x  -    \int_{I_\lambda(\lambda z)} \varphi (x) \d x \Big| \\
		& = \frac{1}{\lambda^d}   \int_{I_\lambda(\lambda (z))} \big| \varphi (x+\lambda e_i)  - \varphi (x) \big| \d x   \\
		& \leq  \frac{1}{\lambda^{d-1}}   \int_{I_\lambda(\lambda (z))}  \int_0^1 \big| \de_i \varphi (x+t \lambda e_i)  \big| \d t  \d x \\
		& =  \frac{1}{\lambda^{d-1}}  \int_0^1 \int_{I_\lambda(\lambda (z))}  \big| \de_i u (x+t \lambda e_i)  \big|   \d x \d t \, .
	\end{split}
\end{equation}
Using the fact that $\Omega^\lambda \subset \subset \Omega'$, for $\lambda$ small enough we obtain
\begin{equation*}
	\begin{split}
		  \sum_{i=1}^d  \sum_{z \in \mathcal{G}_i(\lambda)} \lambda^{d-1} \geo\big(u_\lambda(\lambda(z+e_i)), u_\lambda(\lambda z)\big)  & \leq   \sum_{i=1}^d  \sum_{z \in \mathcal{G}_i(\lambda)}  \int_0^1 \int_{I_\lambda(\lambda (z))}  \big| \de_i u (x+t \lambda e_i)  \big|   \d x \d t \\
		  & \leq  \int_0^1  \sum_{i=1}^d    \int_{\Omega^\lambda}  \big| \de_i u (x+t \lambda e_i)  \big|   \d x \d t \\
		  & \leq    \sum_{i=1}^d    \int_{\Omega'}  \big| \de_i u (x)  \big|   \d x =  \int_{\Omega'}  \big| \nabla u (x)  \big|_{2,1}   \d x \, .
	\end{split}
\end{equation*}

Let $z \in \mathcal{B}_i(\lambda)$. Since $\ol{I_{\lambda}(\lambda z)} \cap \Sigma \neq \emptyset$  or $\ol{I_{\lambda}(\lambda (z+e_i))}  \cap \Sigma \neq \emptyset$, we have that $\ol{I_{\lambda}(\lambda z)} \subset B_{4\lambda\sqrt{d}}(\Sigma)$. By~\cite[Theorem~2.104]{AmbFusPal}, the Minkowski content of $\Sigma$ equals its Hausdorff measure, namely $\frac{\mathcal{L}^d(B_\rho(\Sigma))}{\omega_2 \rho^2} \to \H^{d-2}(\Sigma)$ as $\rho \to 0$.  This implies that 
\begin{equation*}
	\# \mathcal{B}_i(\lambda) \leq \frac{1}{\lambda^d} \mathcal{L}^d(B_{4\lambda\sqrt{d}}(\Sigma)) \leq \frac{1}{\lambda^d} 2 \H^{d-2}(\Sigma) \omega_2 (4\lambda\sqrt{d})^2 \leq C_{\Sigma, d} \frac{1}{\lambda^{d-2}}
\end{equation*}
for $\lambda$ small enough. Using the rough estimate $\geo\big(u_\lambda(\lambda(z+e_i)), u_\lambda(\lambda z)\big) \leq \pi$ we deduce that 
\begin{equation} \label{eq:estimate on bad cubes}
	\begin{split}
		  \sum_{i=1}^d  \sum_{z \in \mathcal{B}_i(\lambda)} \lambda^{d-1} \geo\big(u_\lambda(\lambda(z+e_i)), u_\lambda(\lambda z)\big)   \leq  C_{\Sigma, d}  \lambda \, ,
	\end{split}
\end{equation}
the constant $C_{\Sigma,d}$ being larger than the previous one. 

From~\eqref{eq:estimate on good cubes} and \eqref{eq:estimate on bad cubes} it follows that 
	\begin{equation*}
		\begin{split}
			& \int_{\Omega^{\lambda} \cap J_{u_\lambda} }{ \! \geo(u_\lambda^+,u_\lambda^-)|\nu_{u_\lambda}|_1}{\d \H^{d-1}}  \\
			& \leq \sum_{i=1}^d  \Big( \sum_{z \in \mathcal{G}_i(\lambda)} \lambda^{d-1} \geo\big(u_\lambda(\lambda(z+e_i)), u_\lambda(\lambda z)\big) +   \hspace{-0.8em} \sum_{z \in \mathcal{B}_i(\lambda)} \lambda^{d-1} \geo\big(u_\lambda(\lambda(z+e_i)), u_\lambda(\lambda z)\big) \Big) \\
			& \leq \int_{\Omega'}  \big| \nabla u   \big|_{2,1}   \d x +  C_{\Sigma, d}  \lambda 
		\end{split}
	\end{equation*}
and hence, letting $\lambda \to 0$ and $\Omega' \searrow \Omega$, \eqref{eq:limsup on W11}. Thanks to this step, it suffices to prove the upper bound assuming that the  $\SS^1$-valued map  is constant on each of the cubes $I_\lambda(\lambda z) \subset \Omega^\lambda$. 

\step{5}\ (Construction of $u_N$). Let $u_\lambda \colon \Omega^\lambda \to \SS^1$ be a map that is constant on each of the cubes $I_\lambda(\lambda z)$. We consider the discretization map $\mathfrak{P}_N \colon \SS^1 \to \S_N$ defined as follows: given $a \in \SS^1$, we let $\varphi_a \in [0,2\pi)$ be the unique angle such that $a = \exp(\iota \varphi)$ and we set 
\begin{equation*}
	\mathfrak{P}_N(a) := \exp\big( \iota \theta_N \big\lfloor  \varphi_a/\theta_N \big\rfloor \big) \, .
\end{equation*}
Note that $\geo(\mathfrak{P}_N(a), a) = |\theta_N \big\lfloor  \varphi_a/\theta_N \big\rfloor  - \varphi_a| \leq \theta_N$. We put $u_N := \mathfrak{P}_N(u_\lambda) \in BV(\Omega;\S_N)$.

Then, by the triangle inequality
\begin{equation*}
	\begin{split}
		& \int_{\Omega\cap J_{u_N}}\geo(u_N^-,u_N^+)|\nu_{u_N}|_1\,\mathrm{d}\mathcal{H}^{d-1} \\
		& \quad \leq \sum_{z \in \mathcal{Z}_i(\lambda) } \lambda^{d-1} \geo(u_N(\lambda(z + e_i)),u_N(\lambda z)) \\
		& \quad \leq \int_{\Omega^{\lambda} \cap J_{u_\lambda} }{ \! \geo(u_\lambda^+,u_\lambda^-)|\nu_{u_\lambda}|_1}{\d \H^{d-1}}  \\
		& \quad \quad + \sum_{z \in \mathcal{Z}_i(\lambda) } \lambda^{d-1} \Big( \geo\big(u_N(\lambda(z + e_i)\big),u_\lambda(\lambda (z + e_i))\big) + \geo\big(u_N(\lambda z ),u_\lambda(\lambda  z ) \big) \Big) \\
		& \quad \leq  \int_{\Omega^{\lambda} \cap J_{u_\lambda} }{ \! \geo(u_\lambda^+,u_\lambda^-)|\nu_{u_\lambda}|_1}{\d \H^{d-1}} + \sum_{z \in \mathcal{Z}_i(\lambda) } \lambda^{d-1} 2 \theta_N \\
		& \quad \leq \int_{\Omega^{\lambda} \cap J_{u_\lambda} }{ \! \geo(u_\lambda^+,u_\lambda^-)|\nu_{u_\lambda}|_1}{\d \H^{d-1}} + 2 \theta_N \H^{d-1}({\Omega^\lambda \cap J_{u_\lambda}})\, .
	\end{split}
\end{equation*}
Letting $N \to +\infty$ and by~\eqref{eq:limittheta} we conclude the proof.



  
\end{proof}
\section{Constrained problems} \label{sec:constrained problems}
In this final section we apply the results for the discrete-to-continuum limit to some constrained minimization problem. Again here we can use the more abstract results of \cite{Bra-Cic-Ruf}. We consider the case of discrete Dirichlet boundary conditions and discrete phase constraints. We start with the latter. Note that in both cases we do not state separately the convergence of minimizers which is a standard consequence of the general theory of $\Gamma$-convergence.

\vspace*{0.2cm}
\noindent{\bf Volume constraints in the $N$-clock model:}
Let $V\in (0,1)^N$ be such that $\sum_{k=1}^NV_k=1$. We define a new set of constrained spin configurations by
\begin{equation*}
\mathcal{PC}_{\e}(V):=\left\{u:\e\Z^d\cap \Omega\to\S_N:\,\frac{\#\{u=\exp(ik\theta)\}}{\#(\e\Z^d\cap \Omega)}=V_{k,\e}\quad\forall 1\leq k\leq N\right\}
\end{equation*}
and assume that 
\begin{equation}\label{eq:limitphase}
\lim_{\e\to 0}V_{k,\e}=V_k\quad\quad\forall 1\leq k\leq N.
\end{equation}
Define then the constrained functional
\begin{equation*}
E_{\e,V}^N(u)=\begin{cases}
E_{\e}^N(u) &\mbox{if $u\in\mathcal{PC}_{\e}(V)$} \, , \\
+\infty &\mbox{otherwise in $L^1(\Omega;\R^2)$} \, .
\end{cases}
\end{equation*}
Then by \cite[Theorem 6.2]{Bra-Cic-Ruf} we have the following $\Gamma$-convergence result.
\begin{corollary}\label{c.volconstraint}
Let $N\in\N$ and for $1\leq k\leq N$ let $V_{k,\e}\in (0,1)$ satisfy \eqref{eq:limitphase}. Then as $\e\to 0$ the sequence of functionals $E_{\e,V}^N$ $\Gamma$-converge with respect to the strong $L^1(\Omega;\R^2)$ to the functional $E_{N,V} \colon L^1(\Omega;\R^2)\to [0,+\infty]$ defined by
\begin{equation*}
E_{N,V}(u):=
\begin{cases}
\displaystyle\int_{\Omega\cap J_u}\geo(u^-,u^+)|\nu_u|_1\,\mathrm{d}\mathcal{H}^{d-1} &\mbox{if $u\in BV(\Omega;\S_N)$ and}
\vspace*{-0.3cm}
\\
&\mbox{$|\{u=\exp(ik\theta)\}|=V_k\quad\forall 1\leq k\leq N$,}
\vspace*{0.2cm}
\\
+\infty &\mbox{otherwise.}
\end{cases}
\end{equation*}
\end{corollary}

\noindent {\bf Dirichlet Boundary conditions:}
In order to define discrete Dirichlet boundary conditions and to derive a convergence result, we need to assume some-well preparedness of the boundary condition. For the sake of simplicity we assume that $u_0\in BV_{\rm loc}(\R^d,\S_N)$ is a polyhedral partition such that
\begin{equation}\label{eq:nojumpatbdry}
\mathcal{H}^{d-1}(\Omega\cap J_{u_0})=0.
\end{equation}
We define the set of configurations satisfying a discrete Dirichlet boundary condition $u=u_0$ by
\begin{equation*}
\mathcal{PC}_{\e,u_0}=\left\{u:\e\Z^d\cap\Omega\to\S_N:\,u(\e i)=u_0(\e i)\text{ if dist}(\e i,\partial\Omega)\leq 2\e\right\}.
\end{equation*}
As for the case of volume constraints we define the constrained functionals
\begin{equation*}
E_{\e,u_0}^N(u):=\begin{cases}
E_{\e}^N(u) &\mbox{if $u\in\mathcal{PC}_{\e,u_0}$} \, ,
\\
+\infty &\mbox{otherwise in $L^1(\Omega;\R^2)$} \, .
\end{cases}
\end{equation*}
Since the $\Gamma$-limit result for the sequence $E_{\e}^N$ remains unchanged for any set $\Omega'\supset\supset\Omega$ we can apply \cite[Theorem 4.1 \& Remark 4.2 (i)]{Bra-Cic-Ruf} to obtain the following corollary.
\begin{corollary}
Let $u_0\in BV_{\rm loc}(\R^d;\S_N)$ be a polyhedral partition satisfying \eqref{eq:nojumpatbdry}. Then as $\e\to 0$ the sequence of functionals $E_{\e,V}^N$ $\Gamma$-converge with respect to the strong $L^1(\Omega;\R^2)$ to the functional $E_{N,u_0} \colon L^1(\Omega;\R^2)\to [0,+\infty]$ defined by
\begin{equation*}
E_{N,u_0}(u):=
\begin{cases}
\displaystyle\int_{\Omega\cap J_u}\geo(u^-,u^+)|\nu_u|_1\,\mathrm{d}\mathcal{H}^{d-1}+\int_{\partial\Omega}\geo(u^-,u_0^+)|\nu_x|_1\,\mathrm{d}\mathcal{H}^{d-1} &\mbox{if $u\in BV(\Omega;\S_N)$} \, ,
\\
+\infty &\mbox{otherwise,}
\end{cases}
\end{equation*}
where $\nu_x$ denotes the unit outer normal vector at $\mathcal{H}^{d-1}$-a.e. $x\in\partial\Omega$.
\end{corollary}

\noindent {\bf Acknowledgments.}  The work of M.\ Cicalese was supported by the DFG Collaborative Research Center TRR 109, ``Discretization in Geometry and Dynamics''. G.\ Orlando has received funding from the European Union’s Horizon 2020 research and innovation programme under the Marie Sk\l odowska-Curie grant agreement No 792583.
\bigskip

\bibliographystyle{plain}

\begin{thebibliography}{}	\frenchspacing

	\bibitem{AliBraCic06}
	{\sc R. Alicandro, A. Braides, M. Cicalese}. Phase and anti-phase boundaries in binary discrete systems: a variational viewpoint. {\em Netw. Heterog. Media} {\bf 1} (2006), 85--107.

	\bibitem{AliCicRuf15}
	{\sc R. Alicandro, M. Cicalese, M. Ruf}. Domain formation in magnetic polymer composites: an approach via stochastic homogenization. {\em Arch. Ration. Mech. Anal.} {\bf 218} (2015), 945--984.

	\bibitem{AliCicSig12}
	{\sc R. Alicandro, M. Cicalese, L. Sigalotti} Phase transitions in presence of surfactants: from discrete to continuum. {\em Interfaces Free Bound.} {\bf 14} (2012), 65--103.

	\bibitem{ACL} {\sc R.~Alicandro, A.~Corbo Esposito, C.~Leone.}
	\newblock Relaxation in $BV$ of integral functionals defined on Sobolev functions with values in the unit sphere.
	\newblock {\em J. Convex Anal.} {\bf 14} (2007), 69--98.
	
	\bibitem{AliDLGarPon} {\sc R. Alicandro, L. De Luca, A. Garroni, M. Ponsiglione.} Metastability and dynamics of discrete topological singularities in two dimensions: a $\Gamma$-convergence approach. {\em Arch. Ration. Mech. Anal.} {\bf 214} (2014), 269--330.
	
	\bibitem{AliPon} {\sc R. Alicandro, M. Ponsiglione.} Ginzburg-Landau functionals and renormalized energy: a revised $\Gamma$-convergence approach. {\em J. Funct. Anal.} {\bf 266} (2014), 4890--4907.
	
	
	\bibitem{AmbFusPal} {\sc L. Ambrosio, N. Fusco, D. Pallara}. {\em Functions of Bounded Variation and Free Discontinuity Problems}, Clarendon Press Oxford, 2000.
	
	\bibitem{Ber} {\sc V.L. Berezinskii.} Destruction of long range order in one-dimensional and two
	dimensional systems having a continuous symmetry group. I. Classical systems. {\em
	Sov. Phys. JETP} {\bf 32} (1971), 493--500.
	
	\bibitem{Bet}
	{\sc F.~Bethuel}. The approximation problem for Sobolev maps between two manifolds. {\em Acta Math.} {\bf 167} (1991), 153--206. 
	
	\bibitem{BetBreHel} {\sc F. Bethuel, H. Brezis, F. H\'elein.} {\em Ginzburg-Landau vortices.} Progress in Nonlinear Differential Equations and their Applications, 13. Birkh\"auser Boston MA, 1994.
	
	\bibitem{BraCic17}
	{\sc A. Braides, M. Cicalese}. Interfaces, modulated phases and textures in lattice systems. {\em Arch. Ration. Mech. Anal.} {\bf 223} (2017), 977--1017.
	
	\bibitem{Bra-Cic-Ruf} {\sc A. Braides, M. Cicalese, M. Ruf}. Continuum limit and stochastic homogenization of discrete ferromagnetic thin films. {\em Anal. PDE} {\bf 11} (2018), 499-553.
	
	\bibitem{BraKre18}
	{\sc A. Braides, L. Kreutz}. Design of lattice surface energies. {\em Calc. Var. Partial Differential Equations} {\bf 57}:97 (2018).

	\bibitem{BraPia20}
	{\sc A. Braides, A. Piatnitski}. Homogenization of ferromagnetic energies on Poisson random sets in the plane. {\em Preprint} (2020).

	\bibitem{BraPia12}
	{\sc A. Braides, A. Piatnitski}. Variational problems with percolation: dilute spin systems at zero temperature. {\em Journal of Statistical Physics} {\bf 149} (2012), 846--864.
	
	\bibitem{CafDLL05}
	{\sc L.A. Caffarelli, R. de la Llave}. Interfaces of ground states in Ising models with periodic coefficients. {\em J. Stat. Phys.} {\bf 118} (2005), 687--719.

	\bibitem{CanSeg} {\sc  G. Canevari, A. Segatti.} Defects in nematic shells: a $\Gamma$-convergence discrete-to-continuum approach. {\em Arch. Ration. Mech. Anal.} {\bf 229} (2018), 125--186.
	
	\bibitem{CicForOrl19}
	{\sc M. Cicalese, M. Forster, G. Orlando}. Variational analysis of a two-dimensional frustrated spin system: emergence and rigidity of chirality transitions. {\em SIAM J. Math. Anal.} {\bf 51} (2019),  4848--4893.

	\bibitem{CicSol15}
	{\sc M. Cicalese, F. Solombrino}. Frustrated ferromagnetic spin chains: a variational approach to chirality transitions. {\em J. Nonlinear Sci.} {\bf 25} (2015), 291--313.
		
	\bibitem{COR1} {\sc M.~Cicalese, G.~Orlando, M.~Ruf}. Emergence of concentration effects in the variational analysis of the $N$-clock model. {\em Preprint} (2020).
	
	\bibitem{COR2} {\sc M.~Cicalese, G.~Orlando, M.~Ruf}. The $N$-clock model: Variational analysis for fast and slow divergence rates of $N$. {\em In preparation.}

	\bibitem{vEKO}
	{\sc A. C. D. van Enter, C. K\"ulske, A. A. Opoku}. Discrete approximations to vector spin models. 
	{\em J. Phys. A} {\bf 44}:47 (2011).

	\bibitem{FroSpe} {\sc J. Fr\"ohlich, T. Spencer.} The Kosterlitz-Thouless transition in two-dimensional abelian spin systems and the Coulomb gas. {\em Comm. Math. Phys.} {\bf 81} (1981), 527--602.
	
	\bibitem{Kos} {\sc J.M. Kosterlitz.} The critical properties of the two-dimensional xy model. {\em J. Phys. C} {\bf 6} (1973), 1046--1060.
	
	\bibitem{KosTho} {\sc J.M. Kosterlitz, D.J. Thouless.} Ordering, metastability and phase transitions in two-dimensional systems. {\em J. Phys. C} {\bf 6} (1973), 1181--1203.
	
	\bibitem{KO} {\sc C. K\"ulske, A. A. Opoku} Continuous spin mean-field models: limiting kernels and Gibbs properties of local transforms. 
	{\em J. Math. Phys.} {\bf 49} (2008). 
	
	\bibitem{Lic} {\sc M. W. Licht}. Smoothed projections over weakly Lipschitz domains. {\em Math. Comp.} {\bf 88} (2019), 179--210.
	
	
	\bibitem{Luu-Vae} {\sc J. Luukkainen and J. V\"ais\"al\"a}. Elements of Lipschitz topology. {\em Ann. Acad. Sci. Fenn. Ser. A I Math.} {\bf 3} (1977), 85--122.
	
	
	\bibitem{SanSer-book} {\sc E. Sandier, S. Serfaty.} {\em Vortices in the Magnetic Ginzburg-Landau Model.} Progress in Nonlinear Differential Equations and their Applications, 70. Birkh\"auser Boston, Inc., Boston, MA, 2007.
	
	\end{thebibliography}

\end{document}